\documentclass{article}
\usepackage[utf8]{inputenc}
\usepackage[american]{babel}
\usepackage{enumerate}
\usepackage[cmex10]{amsmath}
\usepackage{amsfonts,amssymb,amsthm,amstext,latexsym,paralist}
\usepackage{url}
\usepackage{xspace}
\usepackage[ruled,lined,boxed,commentsnumbered]{algorithm2e}
\usepackage{tikz}
\usepackage{subfig}
\usepackage{array}
\usepackage{verbatim}
\usepackage{times}
\usepackage{todonotes}

\usepackage{fullpage}

\graphicspath{{maple/}}

\theoremstyle{plain}

\newtheorem{lemma}{Lemma}
\newtheorem{theorem}{Theorem}
\newtheorem{proposition}{Proposition}
\newtheorem{corollary}{Corollary}

\theoremstyle{definition}

\newtheorem{conjecture}{Conjecture}

\theoremstyle{remark}

\newcommand{\ema}[1]{\ensuremath{#1}\xspace}
\newcommand{\tps}{\ema{T}}
\newcommand{\wtps}{\ema{T_{\mathit{wc}}}}
\newcommand{\energy}{\ema{E}}
\newcommand{\ex}[1]{\ema{\mathbb{E}(#1)}}
\newcommand{\pro}[1]{\ema{\mathbb{P}(#1)}}

\newcommand{\tlost}{\ema{\tps_{\mathit{lost}}}}

\newcommand{\elost}{\ema{\energy_{\mathit{lost}}}}

\newcommand{\p}{\ema{\text{p}}}

\newcommand{\tc}{\ema{\tps_{C}}}

\newcommand{\ec}{\ema{\energy_{C}}}

\newcommand{\tex}{\ema{\tps_{\mathrm{exec}}}}

\newcommand{\treex}{\ema{\tps_{\mathrm{reexec}}}}

\newcommand{\pf}{\ema{P_{\mathrm{fail}}}}
\newcommand{\pfun}{\ema{P_{\mathrm{fail}}^{1}}}
\newcommand{\pfdeux}{\ema{P_{\mathrm{fail}}^{2}}}

\newcommand{\spe}{\ema{s}}
\newcommand{\spr}{\ema{\sigma}}
\newcommand{\dl}{\ema{D}}

\newboolean{isRR}
\setboolean{isRR}{true}
\newcommand{\RRme}[2]{\ifthenelse{\boolean{isRR}}{#1}{#2}}

\newcommand{\singlef}{\textsc{SingleFail}\xspace}
\newcommand{\multif}{\textsc{MultipleFail}\xspace}

\RRme{\newcommand{\singlec}{\textsc{SingleChunk}\xspace}}
{\newcommand{\singlec}{\textsc{SC}\xspace}}
\RRme{\newcommand{\multic}{\textsc{MultipleChunks}\xspace}}
{\newcommand{\multic}{\textsc{MC}\xspace}}

\RRme{\newcommand{\singles}{\textsc{SingleSpeed}\xspace}}
{\newcommand{\singles}{\textsc{SS}\xspace}}

\RRme{\newcommand{\multis}{\textsc{MultipleSpeeds}\xspace}}
{\newcommand{\multis}{\textsc{MS}\xspace}}

\RRme{\newcommand{\hd}{\textsc{Hard-Deadline}\xspace}}
{\newcommand{\hd}{\textsc{HD}\xspace}}
\RRme{\newcommand{\ed}{\textsc{Expected-Deadline}\xspace}}
{\newcommand{\ed}{\textsc{ED}\xspace}}

\newcommand{\trfail}{\textsc{Transient}\xspace}

\newcommand{\refequ}[1]{Equation~(\ref{eq.#1})}

\title{Energy-aware checkpointing of divisible tasks\\ with soft or hard deadlines} 
\date{\today}

\RRme{
\author{Guillaume Aupy\thanks{Ecole Normale Sup\'erieure de Lyon, France} \and Anne Benoit\footnotemark[1] \and Rami Melhem\thanks{University of Pittsburgh, USA} \and Paul Renaud-Goud\footnotemark[1]~\thanks{LABRI - University of Bordeaux, France} \and Yves Robert\footnotemark[1]~\thanks{University of Tennessee Knoxville, USA} \\ 
\{guillaume.aupy, anne.benoit, yves.robert\}@ens-lyon.fr; \\
paul.renaud-goud@labri.fr; melhem@cs.pitt.edu}
}{
\author{\IEEEauthorblockN{Guillaume Aupy\IEEEauthorrefmark{1},
Anne Benoit\IEEEauthorrefmark{1},
Rami Melhem\IEEEauthorrefmark{3}, 
Paul Renaud-Goud\IEEEauthorrefmark{1}\IEEEauthorrefmark{4} and
Yves Robert\IEEEauthorrefmark{1}\IEEEauthorrefmark{5}}
\IEEEauthorblockA{\IEEEauthorrefmark{1}LIP - ENS Lyon - INRIA, France}
\IEEEauthorblockA{\IEEEauthorrefmark{3}University of Pittsburgh, USA}
\IEEEauthorblockA{\IEEEauthorrefmark{4}LaBri - University of Bordeaux, France}
\IEEEauthorblockA{\IEEEauthorrefmark{5}University of Tennessee, USA}}
}

\begin{document}
\maketitle

\begin{abstract}
In this paper, we aim at minimizing the energy consumption when executing a divisible workload under
a bound on the total execution time, while resilience is provided through checkpointing.
We discuss several variants of this multi-criteria problem. Given the workload,
we need to decide how many chunks to use, what are the sizes of these chunks,
and at which speed each chunk is executed. Furthermore, since a failure may occur
during the execution of a chunk, we also need to decide at which speed a chunk
should be re-executed in the event of a failure. 
The goal is to minimize the expectation 
of the total energy consumption, while enforcing a deadline on the execution time, 
that should be met either in expectation (soft deadline), or in the worst case (hard deadline). 
For each problem instance, we propose either an exact solution, or a function that 
can be optimized numerically. The different models are then compared through
an extensive set of experiments.
\end{abstract} 

\section{Introduction}

Divisible load scheduling has been extensively studied in the 
past years~\cite{Bharadwaj:1996,divisibleload}.  For divisible applications, 
the computational workload can be  divided into an arbitrary number of
chunks, whose sizes can be freely chosen by the user.  
Such applications occur for instance
in the processing of very large data files, e.g.,  signal processing, linear algebra computation, 
or DNA sequencing.
Traditionally, the goal is to minimize the makespan of the application, i.e., 
the total execution time. 

Nowadays, high performance computing is facing a major challenge with 
the increasing frequency of failures~\cite{IESP-Exascale}. 
There is a need to use fault tolerance or resilience mechanisms to ensure 
the efficient progress and correct termination of the applications in the presence of failures.
A well-established method to deal with failures is {\em checkpointing}: a checkpoint
is taken at the end of the execution of each chunk. During the checkpoint, 
we check for the accuracy of the result; if the result is not correct, due to a transient
failure (such as a memory error or software error), the chunk is re-executed. This model with 
transient failures is one of the most used in the literature, 
see for instance~\cite{Zhu04EEM,Degal05SEI}. 

Furthermore, energy-awareness is now recognized as a first-class constraint in the
design of new scheduling algorithms. To help reduce energy
dissipation, current processors from AMD, Intel and Transmetta allow
the speed to be set dynamically, using a dynamic voltage and frequency
scaling technique (DVFS). Indeed, a processor running at speed $s$
dissipates $s^3$ watts per unit of time~\cite{pow3IPDPS}.  
We therefore focus on two objective functions: execution time and
energy consumption, while resilience is ensured through checkpointing. 
More precisely, we aim at minimizing energy consumption, including that of 
checkpointing and re-execution in case of failure,
while enforcing a bound on execution time. 

Given a workload~$W$, we need to decide how many chunks
to use, and of which sizes. Using more chunks leads to a higher checkpoint cost, but smaller chunks 
imply less computation loss (and less re-execution) when a failure occurs. We assume that a chunk can fail only once, 
i.e., we re-execute each chunk at most once. Indeed, the probability that a fault 
would strike during both the first execution and the re-execution is negligible.
\RRme{We discuss the accuracy of this assumption in Section~\ref{sec.compare}. }
{The accuracy of this assumption is discussed in the companion research report~\cite{rr-abmrr}. }

Due to the probabilistic nature of failure hits, it is natural to
study the expectation \ex{\energy} of the energy consumption, because it represents the average cost
over many executions. As for the bound \dl on execution time (the deadline), 
there are two relevant scenarios:
either we enforce that this bound is a \emph{soft deadline} to be met in expectation,
or we enforce that this bound is a \emph{hard deadline} to be
met in the worst case. The former scenario corresponds to flexible environment 
where task deadlines can be viewed as
average response times~\cite{buttazzo2005}, while the latter scenario 
corresponds to real-time environments where task deadlines are always
strictly enforced~\cite{Stankovic:1998:DSR:552538}.
In both scenarios, we have to determine the number of chunks, their sizes,
and the speed at which to execute (and possibly re-execute) every
chunk.

Our first contribution is to formalize this important multi-objective problem. 
The general problem consists of finding $n$, the number of chunks,
as well as the speeds for the execution and the re-execution of each chunk,
both for soft and hard deadlines. We identify and discuss two important sub-cases
that help tackling the most general problem instance: (i)  a single chunk  (the task is atomic);
and (ii) re-execution speed is always identical to the first execution speed.
The second contribution is a comprehensive study of all 
problem instances; for each instance, we propose either an exact solution, or a function
that can be optimized numerically. 
\RRme{
We also analytically prove the accuracy of our model that enforces a single re-execution
per chunk.}{}
We then compare the different models
through an extensive set of experiments. We compare the optimal energy consumption
under various models with a set of different parameters. 
It turns out that when $\lambda$ is small, it is sufficient to restrict the study to a single chunk, 
while when $\lambda$ increases, it is better to use multiple chunks and different
re-execution speeds. 

The rest of the paper is organized as follows.
First we discuss related work in Section~\ref{sec.related}. The model and the optimization problems 
are formalized in Section~\ref{sec.framework}. \RRme{We discuss the accuracy of the model in 
Section~\ref{sec.compare}.}{} We first focus in Section~\ref{sec.singletask} on the simpler case 
of an atomic task, i.e., with a single chunk. 
The general problem with multiple chunks,
where we need to decide for the number of chunks and their sizes, is discussed in Section~\ref{sec.multiple}. 
In Section~\ref{sec.experiments}, we report several experiments to assess the differences between the models, 
and the relative gain due to chunking or to using different speeds for execution and re-execution.
Finally, we provide some concluding remarks and future research directions in Section~\ref{sec.conclusion}.

\section{Related work}
\label{sec.related}

Dynamic power management through voltage/fre\-quen\-cy scaling~\cite{Weiser94schedulingfor} 
utilizes the slack in a given computation to reduce energy consumption 
while checkpointing.The authors of~\cite{Chandy:1985, 9847} utilize that slack to improve the reliability 
of the computation. Hence, it is natural to explore the interplay of power management 
and fault tolerance~\cite{ Melhem03CP}, when both techniques result in delaying 
the completion time of tasks, thus resulting in a tradeoff between power consumption, 
reliability and performance. This tri-criteria optimization problem has been explored 
by many researchers, especially in real-time and embedded systems where the 
completion time of a task is as important as the reliability of its result.

The power/reliability/performance tradeoff has been explored from many different angles. 
In~\cite{ Zhang03CP}, an adaptive scheme is presented to place checkpoints 
based on the expected frequency of faults and is combined with dynamic speed scaling depending 
on the actual occurrence of faults. Similarly, in~\cite{ Melhem03CP}, the placement of 
checkpoints is chosen in a way that minimizes the total energy consumption assuming 
that the slack reserved for rollback recovery is used for speed scaling if faults do not occur. 
In~\cite{Zhu04EEM}, the effect of frequency scaling on the fault rate was considered and 
incorporated into the optimization problem. In~\cite{ Zhu06}, the study of the tri-criteria 
optimization was extended to the case of multiple tasks executing on the same processor. 
In~\cite{pop-poulsen}, a constraint logic programming-based approach is presented 
to decide for the voltage levels, the start times of processes and the transmission times 
of messages, in such a way that transient faults are tolerated, timing constraints are satisfied 
and energy is minimized. 

Recently, off-line scheduling heuristics that consider the three criteria were presented 
for systems where active replication, rather than fault recovery, is used to enhance 
reliability~\cite{ Assayad11}.  Selective re-execution of some tasks were considered 
in~\cite{DBLP:journals/corr/abs-1111-5528} to achieve a given level of reliability while 
minimizing energy, when tasks graphs are scheduled on multiprocessors with hard deadlines. 
Approximation algorithms for particular types of task graphs were presented to efficiently 
solve the same problem in~\cite{aupy:hal-00742754}.

In this work, we consider two types of deadlines that are commonly used for real-time tasks; 
hard and soft deadlines. In hard real-time systems~\cite{Stankovic:1998:DSR:552538}, 
deadlines should be strictly met and any computation that does not meet its deadline 
is not useful to the system. These systems are built to cope with worst-case scenarios, 
especially in critical applications where catastrophic consequences may result from missing 
deadlines. Soft real-time systems~\cite{buttazzo2005} are more flexible and are designed 
to adapt to system changes that may prevent the meeting of the deadline. They are suited 
to novel applications such as multimedia and interactive systems. In these systems, it is 
desired to reduce the expected completion time rather than to meet hard deadlines.

\section{Framework}
\label{sec.framework}

Given a workload~$W$, the problem is to divide $W$ into a number of
chunks and to decide at which speed each chunk is executed. In case
of a transient failure during the execution of one chunk, this chunk is re-executed, possibly at
a different speed. 
\RRme{We formalize the model in Section~\ref{sec.model}, and then different
variants of the optimization problem are defined in Section~\ref{sec.pbs}.  
\RRme{Table~\ref{tab.notations} summarizes the main notations.}
{}

\begin{table}[htdp]
\begin{center}
\begin{tabular}{|l|l|}
\hline
$W$ & total amount of work\\
\spe & processor speed for first execution\\
\spr &  processor speed for re-execution\\
\tc & checkpointing time\\
\ec & energy spent for checkpointing\\
\hline
\end{tabular}
\end{center}
\caption{List of main notations.}
\label{tab.notations}
\end{table}
}{}

\RRme{
\subsection{Model}
\label{sec.model}}{
\smallskip  
{\bf Model. } }Consider first the case of a single chunk (or atomic task) of size~$W$, 
denoted as \singlec. 
We execute this chunk on a processor that can run at several speeds.  We assume
continuous speeds, i.e., the speed of execution can take an
arbitrary positive real value.  The execution is subject
to failure, and resilience is provided through the use of
checkpointing.  The overhead induced by checkpointing is twofold:
execution time \tc, and energy consumption~\ec.

We assume that failures strike with uniform distribution, hence the
probability that a failure occurs during an execution is linearly
proportional to the length of this execution. Consider the first
execution of a task of size $W$ executed at speed $\spe$: the
execution time is $\tex = W/\spe +\tc$, hence the failure probability
is $\pf=\lambda \tex = \lambda(W/\spe +\tc)$, where $\lambda$ is the
instantaneous failure rate.  If there is indeed a failure, we
re-execute the task at speed $\spr$ (which may or may not differ from
$\spe$); the re-execution time is then $\treex = W/\spr +\tc$ so that
the expected execution time is
\begin{align}
\!\!\!\!\! \ex{\tps} & \!=\! \tex + \pf  \treex \nonumber\\
          &\!=\! (W/\spe +\tc) + \lambda (W/\spe +\tc) (W/\spr +\tc) \; .
\label{eq.exptime}
\end{align}
Similarly, the worst-case execution time is
\begin{align}
\wtps &= \tex +  \treex \nonumber\\
          &= (W/\spe +\tc) +  (W/\spr +\tc) \; .
\label{eq.exptimewc}
\end{align}

Remember that we assume success after re-execution, so we do not
account for second and more re-executions. Along the same line, we could
spare the checkpoint after re-executing the last task in a series of
tasks, but this unduly complicates the analysis.
\RRme{In Section~\ref{sec.compare},}
{In the companion research report~\cite{rr-abmrr},} we show that this model with only
a single re-execution is accurate up to second order terms when
compared to the model with an arbitrary number of failures that follows
an Exponential distribution of parameter~$\lambda$.

What is the expected energy
consumed during execution?  The energy consumed during the first
execution at speed $\spe$ is $W \spe ^2 + \ec $, where $\ec$ is the energy
consumed during a checkpoint. The energy
consumed during the second execution at speed $\spr$ is $W \spr ^2 +
\ec $, and this execution takes place with probability $\pf=\lambda
\tex = \lambda(W/\spe +\tc)$, as before.  Hence the expectation of the
energy consumed is
\begin{align}
\!\!\!\ex{\energy} &\!=\! (W \spe^2 \!+\! \ec) \!+\! \lambda \left (W/\spe \!+\! \tc
\right ) \left( W \spr^2 \!+\! \ec \right).
\label{eq.expenergy}
\end{align}

With multiple chunks (\multic model), 
the execution times (worst case or expected) are the sum of the execution times for each chunk, and the expected energy is the sum of the expected energy for each chunk
(by linearity of expectations).
\RRme{

}{} 
We point out that the failure model is coherent with respect to chunking. Indeed, assume
that a divisible task of weight $W$ is split into two chunks of weights $w_{1}$ and
$w_{2}$ (where $w_{1} + w_{2} =W$). Then
the probability of failure for the first chunk is $\pfun =  \lambda (w_{1}/\spe +\tc)$
and that for the second chunk is $\pfdeux =  \lambda (w_{2}/\spe +\tc)$. 
The probability of failure  $\pf =  \lambda (W/\spe +\tc)$ with a single  chunk differs from the
probability of failure with two chunks only because of the
extra checkpoint that is taken; if $\tc = 0$, they coincide exactly. 
If $\tc > 0$, there is an additional risk to use two chunks, because the execution lasts longer by
a duration $\tc$. Of course this is the price to pay for a shorter re-execution time in case of failure:
Equation~\eqref{eq.exptime} shows that the expected re-execution time is 
$\pf  \treex$, which is quadratic in $W$. There is a trade-off between having many small chunks
(many $\tc$ to pay, but small re-execution cost) and a few larger chunks (fewer $\tc$, but increased re-execution cost).

\RRme{
\subsection{Optimization problems}
\label{sec.pbs}

}{
\smallskip
{\bf Optimization problems.}}
\RRme{The optimization problem 
is stated as follows:
g}{G}iven a deadline \dl and a divisible task whose total
computational load is $W$, the problem is to partition the task into $n$
chunks of size $w_{i}$, where $\sum_{i=1}^{n} w_{i} = W$, and choose
for each chunk an execution speed $\spe_{i}$ and a re-execution speed~$\spr_{i}$ 
in order to minimize the expected energy consumption:
\begin{equation*}
\ex{\energy} = \sum_{i=1}^{n}(w_{i} \spe_{i}^2 + \ec) + \lambda \left
  (\frac{w_{i}}{\spe_{i}} + \tc \right ) \left( w_{i} \spr_{i}^2 + \ec
\right),
\end{equation*}
subject to the constraint that the deadline is met either in expectation or in the worst case:
\RRme{
\begin{equation*}
\begin{array}{ll}
\ed & \ex{\tps} = \sum_{i=1}^{n}\left(\frac{w_{i}}{\spe_{i}} + \tc + \lambda
\left (\frac{w_{i}}{\spe_{i}} + \tc \right ) \left(
  \frac{w_{i}}{\spr_{i}}+ \tc \right) \right) \leq \dl\\
 \hd & \wtps = \sum_{i=1}^{n}\left(\frac{w_{i}}{\spe_{i}} + \tc + 
  \frac{w_{i}}{\spr_{i}}+ \tc \right) \leq \dl
\end{array}
\end{equation*}
}{
\begin{equation*}
\begin{array}{l}
\ed \RRme{}{\text{  (Expected deadline): }} \\ \ex{\tps} \!=\! \sum_{i=1}^{n}\!\left(\frac{w_{i}}{\spe_{i}} \!+\! \tc \!+\! \lambda
\left (\frac{w_{i}}{\spe_{i}} \!+\! \tc \right )\! \left(
  \frac{w_{i}}{\spr_{i}}\!+\! \tc \right) \right) \!\leq\! \dl\\[.3cm]
 \hd \RRme{}{ \text{ (Hard deadline): }}\\ \wtps = \sum_{i=1}^{n}\left(\frac{w_{i}}{\spe_{i}} + \tc + 
  \frac{w_{i}}{\spr_{i}}+ \tc \right) \leq \dl
\end{array}
\end{equation*}
}
The unknowns are the number of chunks~$n$, the sizes of these
chunks~$w_{i}$, the speeds for the first execution~$\spe_{i}$ and 
\RRme{the speeds }{} for the second execution~$\spr_{i}$.  
We consider two variants of the problem, depending upon re-execution speeds:
\begin{itemize}
\item \singles\RRme{}{(Single speed)}: in this simpler variant, the re-execution speed is
  always the same as the speed chosen for the first execution.  We
  then have to determine a single speed for each chunk: $\spr_{i}
= \spe_{i}$ for all $i$. 
\item \multis\RRme{}{(Multiple speeds)}: in this more general variant, the re-execution speed is
  freely chosen, and there are two different speeds to determine for
  each chunk.
\end{itemize}

We also consider the variant with a single chunk (\singlec), i.e., the task is atomic
and we only need to decide for its execution speed (in the \singles model),
or for its execution and re-execution speeds (in the \multis model). We
start the study in Section~\ref{sec.singletask} with this simpler problem.

 \RRme{
\section{Accuracy of the model}
\label{sec.compare}

In this section, we discuss the accuracy of this model, which accounts for a single re-execution.
We compare the expressions of the expected deadline and energy (in Equations~\eqref{eq.exptime}
and~\eqref{eq.expenergy}) to those obtained when
adopting the more advanced model where an arbitrary number of
Exponentially distributed failures can strike during execution and
re-execution. We only deal with soft deadlines here, because no hard deadline can be enforced
for the model with Exponentially distributed failures (the execution time of a chunk can be arbitrarily
large, although such an event has low probability to occur).

Assume that failures are distributed using an Exponential distribution
of parameter $\lambda$: the probability of failure during a time
interval of length $t$ is $\pf = 1 - e^{-\lambda t}$.  Consider a
single task of size $W$ that we first execute at speed $\spe$. If we
detect a transient failure at the end of the execution, we re-execute
the task until success, using speed $\spr$ at each of these new
attempts.  To the best of our knowledge, the expressions for $\ex{\tps}$
and $\ex{\energy}$ are unknown for this model, and we establish them below:

\begin{proposition}
\label{prop.indiv.exp}
With an arbitrary number of Exponentially distributed failures and one
single task of size~$W$, 
\RRme{
\begin{align}
 \label{eq.timeexpo}
\ex{\tps} &= W / \spe + \tc  \quad + e^{ \lambda (W/\spr + \tc)} \left(1  -e^{-\lambda (W/\spe +\tc)} \right)
\left( W/\spr + \tc \right)\\
 \label{eq.energyexpo}
\ex{\energy} &= W \spe^2 + \ec \quad + e^{ \lambda (W/\spr + \tc)} \left(1  -e^{-\lambda (W/\spe +\tc)} \right)
\left( W \spr^2 + \ec \right)
\end{align}
}{
\begin{align}
 \label{eq.timeexpo}
\ex{\tps} &= W / \spe + \tc  \\
& \quad + e^{ \lambda (W/\spr + \tc)} \left(1  -e^{-\lambda (W/\spe +\tc)} \right)
\left( W/\spr + \tc
 \right) \nonumber\\
 \label{eq.energyexpo}
\ex{\energy} &= W \spe^2 + \ec  \\
& \quad + e^{ \lambda (W/\spr + \tc)} \left(1  -e^{-\lambda (W/\spe +\tc)} \right)
\left( W \spr^2 + \ec \right)\nonumber
\end{align}
}
\end{proposition}
\begin{proof}
  With an Exponential distribution, \refequ{exptime} can be rewritten as
  $\ex{\tps} = \tex
  + \pf \ex{\treex}$,
  where $\tex = W/s + \tc$ and $\pf = 1 - e^{-
    \lambda (W/s + \tc)}$. Since all re-executions are done at speed
  $\sigma$, the
  expectation of the re-execution time obeys the following equation:
  $$\ex{\treex} = (W/\sigma + \tc) + \left( 1- e^{-\lambda(W/\sigma +
      \tc)}\right) \ex{\treex}$$
   We use the memoryless property of the Exponential distribution here: after a failure, the
   expectation of the time to re-execute the task is exactly the same as before the failure
        This leads to $\ex{\treex} =
  e^{\lambda(W/\sigma + \tc)}(W/\sigma + \tc)$. Reporting in the first
  equation, we end up with \refequ{timeexpo}.
  The expression of the expected energy consumption
  (\refequ{energyexpo}) is derived using the same line of reasoning.
\end{proof}


\begin{proposition}
\label{prop.accurate}
With an arbitrary number of Exponentially distributed failures and one
single task of size~$W$, when $\lambda \rightarrow 0$, 
\begin{equation}
\label{eq.devtime}
\ex{\tps} = (W/\spe + \tc) + \lambda (W/\spe + \tc)(W/\sigma + \tc) + O\left( \lambda^2 \right)
\end{equation}
\begin{equation}
\label{eq.devenergy}
\ex{\energy} = (W \spe^2 + \ec) + \lambda (W/\spe + \tc)(W \sigma^2 + \ec) + O\left( \lambda^2 \right)
\end{equation}
\end{proposition}

\begin{proof}
The first-order Taylor expansion of $x \mapsto e^x$ around $0$ gives:
\RRme{
\begin{align*}
\ex{\tps} \underset{\lambda \rightarrow 0}{=} (W/\spe + \tc)
+& \left( 1 + \lambda (W/\spe + \tc) + O \left( \lambda^2 (W/\spe + \tc)^2 \right) \right)\\
&\times \left( \lambda (W/\sigma + \tc) + O \left( \lambda^2 (W/\sigma + \tc)^2 \right) \right)
\left( W/\sigma + \tc \right)
\end{align*}
}{
\begin{multline*}
\ex{\tps} \underset{\lambda \rightarrow 0}{=} (W/\spe + \tc)
+ \left( 1 + \lambda (W/\spe + \tc) + O \left( \lambda^2 (W/\spe + \tc)^2 \right) \right)\\
\times \left( \lambda (W/\sigma + \tc) + O \left( \lambda^2 (W/\sigma + \tc)^2 \right) \right)
\left( W/\sigma + \tc \right)
\end{multline*}
}
Hence,
\begin{align*}
\ex{\tps} &\underset{\lambda \rightarrow 0}{=} (W/\spe + \tc)
+\left( \lambda (W/\spe + \tc) + O \left( \lambda^2 \right) \right)
\left( W/\sigma + \tc \right)\\
\ex{\tps} &\underset{\lambda \rightarrow 0}{=} (W/\spe + \tc) + \lambda (W/\spe + \tc)(W/\sigma + \tc) + O\left( \lambda^2 \right)
\end{align*}
Again, the energy formula is built using the same rationale.
\end{proof}

As a consequence of Proposition~\ref{prop.accurate}, the formulas that we
consider with one single re-execution (Equations~\eqref{eq.exptime}
and~\eqref{eq.expenergy}) are accurate up to second order
terms when compared to the model with an arbitrary number of Exponential failures.  
Note that this result is not obvious, because we drop a
potentially arbitrarily large number of re-executions in the linear
model with at most one re-execution. 
Furthermore, the result extends naturally when considering a divisible task
and \multic, since the result holds for each chunk, and by summation,
one single re-execution of each chunk is accurate up to second order terms. 
}{}

\section{With a single chunk}
\label{sec.singletask}

In this section, we consider the case of a single chunk, or equivalently of an atomic task:
given a non-divisible workload~$W$ and a deadline
\dl, find the values of $\spe$ and $\spr$ that minimize 
$$\ex{\energy} =
(W \spe^2 + \ec) + \lambda \left (\frac{W}{\spe} + \tc \right) \left(
  W \spr^2 + \ec \right)$$
   subject to 
   $$\ex{\tps} = \left(
  \frac{W}{\spe} + \tc \right) + \lambda \left (\frac{W}{\spe} + \tc
\right) \left( \frac{W}{\spr} + \tc \right) \leq \dl$$ in the \ed model, and subject to
 $$ \frac{W}{\spe} + \tc + \frac{W}{\spr} + \tc \leq \dl$$ in the \hd model.
We first deal
with the \singles model, where we enforce $\spr = \spe$, before
moving on to the \multis model.

\subsection{Single speed model}
\label{sec.singles-one}
In this section, we express $\ex{\energy}$ as functions of the speed \spe. That is, 
$\ex{\energy}(\spe) = (W \spe^2 + \ec)(1 + \lambda (W/\spe + \tc))$.
The following result is valid for both \ed and \hd models.

\RRme{
\begin{lemma}
	\label{g.singles}
$\ex{\energy}$  is convex on $\mathbb{R^{\star}_+}$. It admits a unique minimum
\RRme{
\begin{equation}
	\label{eq.spestar}
\spe^{\star} = \frac{\lambda W}{6(1+ \lambda \tc)} \left(\frac{-(3 \sqrt{3} \sqrt{27 a^2-4 a}-27 a+2)^{1/3}}{2^{1/3}} -\frac{2^{1/3}}{(3 \sqrt{3} \sqrt{27 a^2-4 a}-27 a+2)^{1/3}}-1 \right )
\end{equation}
}{
\begin{align}
	\label{eq.spestar}
\spe^{\star} &= \frac{\lambda W}{6(1+ \lambda \tc)} \left(\frac{-(3 \sqrt{3} \sqrt{27 a^2-4 a}-27 a+2)^{1/3}}{2^{1/3}} \right . \nonumber \\
& \quad \left . -\frac{2^{1/3}}{(3 \sqrt{3} \sqrt{27 a^2-4 a}-27 a+2)^{1/3}}-1 \right )
\end{align}}
where $a=\lambda \ec \left (\frac{2(1+ \lambda \tc)}{\lambda W}\right )^2$.
\end{lemma}

\begin{proof}
Let us prove that $g(\spe)=\ex{\energy}(s)$ is convex and admits a unique minimum:
we have $g'(\spe)=\spe(2W (1+ \lambda \tc)) + \lambda W^2 - \frac{\lambda W\ec}{\spe^2}$, 
$g''(\spe)=(2W (1+ \lambda \tc))+ \frac{2\lambda W\ec}{\spe^3}>0$. This function is strictly 
convex in $\mathbb{R^{\star}_+}$, and $g'  \underset{0^+}{\rightarrow} -\infty$, 
$g' \underset{\infty}{\rightarrow} \infty$  thus there exist a unique minimum.

Let us find the minimum.  
For $\spe>0$, we have:
\RRme{
\begin{align*}
g'(\spe)=0 &\Leftrightarrow \left (\frac{2(1+ \lambda \tc)}{\lambda W}\right )^3\spe^3 + \left (\frac{2(1+ \lambda \tc)}{\lambda W}\right )^2 \spe^2  - \lambda \ec \left (\frac{2(1+ \lambda \tc)}{\lambda W}\right )^2 =0\\
&\Leftrightarrow X^3 + X^2 - \lambda \ec \left (\frac{2(1+ \lambda \tc)}{\lambda W}\right )^2 =0 \quad \text{ where } X=\frac{2(1+ \lambda \tc)}{\lambda W}\spe
\end{align*}
}{
\begin{align*}
g'(\spe)=0 &\Leftrightarrow \left (\frac{2(1+ \lambda \tc)}{\lambda W}\right )^3\spe^3 + \left (\frac{2(1+ \lambda \tc)}{\lambda W}\right )^2 \spe^2 \\
&\quad \quad - \lambda \ec \left (\frac{2(1+ \lambda \tc)}{\lambda W}\right )^2 =0\\
&\Leftrightarrow X^3 + X^2 - \lambda \ec \left (\frac{2(1+ \lambda \tc)}{\lambda W}\right )^2 =0 \\
& \quad \quad \text{ where } X=\frac{2(1+ \lambda \tc)}{\lambda W}\spe
\end{align*}}
Using a computer algebra software, it is easy to show that the minimum is obtained at the value $\spe = \spe^{\star}$
given by Equation~\ref{eq.spestar}.
\end{proof}
}{
\begin{lemma}
	\label{g.singles}
$\ex{\energy}$  is convex on $\mathbb{R^{\star}_+}$. It admits a unique minimum 
$\spe^{\star}$ which can be computed numerically.
\end{lemma}
Due to lack of space, the proof is available in the companion research report~\cite{rr-abmrr}.
}

\medskip
\subsubsection{Expected deadline}
In the \singles \ed model, we denote $\ex{\tps}(\spe) = (W/\spe + \tc)(1 + \lambda (W/\spe + \tc))$ 
the constraint on the execution time.

\begin{lemma}
	\label{f.singles}
For any $\dl$, if $\tc + \lambda \tc^2\geq \dl$, then there is no solution. Otherwise, the constraint 
on the execution time can be rewritten as
$\spe \in \left [W\dfrac{1+2\lambda \tc + \sqrt{4 \lambda \dl +1}}{2(\dl -\tc(1+\lambda \tc))}, +\infty \right ($.
\end{lemma}

\begin{proof}
The function $s \mapsto \ex{\tps}(\spe)$ is strictly decreasing and converges to $\tc + \lambda \tc^2$.
Hence, if $\tc + \lambda \tc^2\geq \dl$, then there is no solution. Else there exist a minimum speed $\spe_0$
such that, $\ex{\tps}(\spe_0)=\dl$, and for all $\spe \geq \spe_0$, $\ex{\tps}(\spe)\leq \dl$.

More precisely, $\spe_0 = W\dfrac{1+2\lambda \tc + \sqrt{4 \lambda \dl +1}}{2(\dl -\tc(1+\lambda \tc))}$:
since there is a unique solution to $\ex{\tps}(\spe) = \dl$, we can solve this equation in order to 
find $\spe_0$.
\end{proof}

To simplify the following results, we define
\begin{equation}
\spe_0 = W\dfrac{1+2\lambda \tc + \sqrt{4 \lambda \dl +1}}{2(\dl -\tc(1+\lambda \tc))}.
\end{equation}
\begin{proposition}
In the \singles model, it is possible to numerically compute the optimal solution for \singlec as follows:
\begin{enumerate}
	\item If $\tc + \lambda \tc^2\geq \dl$, then there is no solution;
	\item Else, the optimal speed is $\max (\spe_0,\spe^{\star})$.
\end{enumerate}
\end{proposition}

\begin{proof}
This is a corollary of Lemma~\ref{g.singles}: because $s \mapsto \ex{\tps}(\spe)$ is convex on $\mathbb{R^{\star}_+}$, 
then its restriction to the interval $[\spe_0,+\infty( $ is also convex and admits a unique minimum:
\begin{itemize}
	\item if $\spe^{\star}<\spe_0$, then \ex{\tps}(\spe) is increasing on $[\spe_0,+\infty($, then the optimal solution is $\spe_0$
	\item else, clearly the minimum is reached when $\spe=\spe^{\star}$.
\end{itemize}
The optimal solution is then $\max (\spe_0,\spe^{\star})$.
\end{proof}

\medskip
\subsubsection{Hard deadline}
In the \hd model, the bound on the execution time can be written as $2\left (\frac{W}{\spe} + \tc \right) \leq \dl$

\begin{lemma}
	\label{f.singles.hd}
In the \singles \hd model, for any $\dl$, if $2\tc\geq \dl$, then there is no solution. Otherwise, the 
constraint on the execution time can be rewritten as
$\spe \in \left [\frac{W}{\frac{\dl}{2}-\tc}; + \infty \right ($
\end{lemma}
\begin{proof}
The constraint on the execution time is now $2\left (\frac{W}{\spe} + \tc \right) \leq \dl$.
\end{proof}

\begin{proposition}
Let $\spe^{\star}$ be the solution indicated in \RRme{Equation~\ref{eq.spestar}}{Lemma~\ref{g.singles}}.
In the \singles \hd model if $2\tc\geq \dl$, then there is no solution. Otherwise, the minimum is 
reached when $\spe=\max \left (\spe^{\star},\frac{W}{\frac{\dl}{2}-\tc} \right )$.
\end{proposition}
\begin{proof}
The fact that there is no solution when  $2\tc\geq \dl$ comes from Lemma~\ref{f.singles.hd}.
Otherwise, the result is obvious by convexity of the expected energy function.
\end{proof}

\subsection{Multiple speeds model}
\label{sec.multis-one}
In this section, we consider the general \multis model. We use the following notations:
\[ \ex{\energy}(\spe,\sigma) = (W \spe^2 + \ec) + \lambda (W/\spe + \tc)(W \spr^2 + \ec) \]

\RRme{
Let us first introduce a preliminary Lemma:
\begin{lemma}[Convexity \singlec]
	\label{conv.tight.undiv}
The problem of minimizing $A_0 + \alpha_0 x^2$ under the constraint $A_1 + \frac{\alpha_1}{x} \leq A_2$
where $A_0, A_1, A_2$ are constants and $\alpha_0, \alpha_1$ are positive constants is solved 
when $x$ is minimum, that is when $A_1 + \frac{\alpha_1}{x} = A_2$.
\end{lemma}

\begin{proof}
The function $A_0 + \alpha_0 x^2$ is strictly increasing, so it is is minimized when $x$ is minimum.
The function $A_1 + \frac{\alpha_1}{x}$ is strictly decreasing with $\lim_{x\rightarrow 0}= +\infty$, 
hence an upper bound is reached when $x$ is minimum. With those two results, we can say that the 
constraint should be tight in order to solve our problem.
\end{proof}
}{}

\subsubsection{Expected deadline}
The execution time in the \multis \ed model can be written as
$$\ex{\tps}(\spe,\spr) = (W/\spe + \tc) + \lambda (W/\spe + \tc)(W/\spr + \tc)$$
We start by giving a useful property, namely that the deadline is always tight in the \multis \ed model:

\begin{lemma}
In the \multis \ed model, in order to minimize the energy consumption, the deadline should be tight.
\end{lemma}
\RRme{
\begin{proof}
Considering \spe and $W$ fixed, then
$\ex{\tps}(\spe,\spr) = \tps_0 + \frac{\alpha}{\spr} \leq \dl$, and 
$\ex{\energy}(\spe,\spr) = \energy_0 + \alpha \spr^2$,
where $\tps_0 = (W / \spe + \ec) + \lambda \tc (W/\spe + \tc) $,
$\energy_0 = (W \spe^2 + \ec) + \lambda \ec (W/\spe + \tc) $
and $\alpha = W (W/\spe + \tc)$ are constant.
With Lemma~\ref{conv.tight.undiv} we conclude that the deadline should be tight.
\end{proof}
}{Due to lack of space, the proof is available in the companion research report~\cite{rr-abmrr}.}


This lemma allows us to express $\spr$ as a function of $\spe$:
\[ \spr = \frac{\lambda W}{\frac{\dl}{ \frac{W}{\spe} + \tc} - (1 + \lambda \tc)}. \]
Also we reduce the bi-criteria problem to the minimization problem of the single-variable function:
\RRme{
\begin{equation}
	\label{eq.multis.undiv}
 \spe \mapsto W \spe^2 +\ec + \lambda \left(  \frac{W}{\spe} + \tc\right ) \left (W \left (\frac{\lambda W}{\frac{\dl}{ \frac{W}{\spe} + \tc} - (1 + \lambda \tc)} \right )^2 + \ec \right) 
 \end{equation}
 }{
 \begin{align}
	\label{eq.multis.undiv}
 \spe \mapsto& W \spe^2 +\ec + \lambda \left(  \frac{W}{\spe} + \tc\right ) \times  \\
 & \quad \left (W \left (\frac{\lambda W}{\frac{\dl}{ \frac{W}{\spe} + \tc} - (1 + \lambda \tc)} \right )^2 + \ec \right)  \nonumber
 \end{align}}
which can be solved numerically.

\medskip
\subsubsection{Hard deadline}

In this model we have similar results as with \ed. The constraint on the execution time 
writes: $\frac{W}{\spe} + \tc + \frac{W}{\spr} + \tc \leq \dl$.
\RRme{
Another corollary of Lemma~\ref{conv.tight.undiv} is:
\begin{lemma}
In the \multis \ed model, in order to minimize the energy consumption, the deadline should be tight.
\end{lemma}
}{
\begin{lemma}
In the \multis \ed model, in order to minimize the energy consumption, the deadline should be tight.
\end{lemma}
Again, the proof can be found in the companion research report~\cite{rr-abmrr}.
}

This lemma allows us to express \spr as a function of \spe:
\[
\spr = \frac{W}{(\dl-2\tc)\spe - W} \spe
\]
Finally, we reduce the bi-criteria problem to the minimization problem of the single-variable function:
\RRme{
\begin{equation}
	\label{eq.multis.undiv.hd}
 \spe \mapsto W \spe^2 +\ec + \lambda \left(  \frac{W}{\spe} + \tc\right ) \left (W \left (\frac{W}{(\dl-2\tc)\spe - W} \spe \right )^2 + \ec \right)
 \end{equation}
 }{
  \begin{align}
	\label{eq.multis.undiv.hd}
 \spe \mapsto& W \spe^2 +\ec + \lambda \left(  \frac{W}{\spe} + \tc\right ) \times  \\
 & \quad \left (W \left (\frac{W}{(\dl-2\tc)\spe - W} \spe \right )^2 + \ec \right)  \nonumber
 \end{align}
 }
which can be solved numerically.

\section{Several chunks}
\label{sec.multiple}

 In this section, we deal with the general problem of a divisible task of size $W$ 
 that can be split into an arbitrary number of chunks.
 We divide the task into $n$ chunks of size $w_i$ such that
 $\sum_{i=1}^n w_i = W$. Each chunk is executed once at speed
 $\spe_i$, and re-executed (if necessary) at speed $\spr_i$.  
The problem is to find the values of $n$, $w_i$, $\spe_i$ and $\spr_i$ that minimize 
\RRme{
\[\ex{\energy} =
\sum_i \left (w_i \spe_i^2 + \ec \right ) + \lambda \sum_i  \left (\frac{w_i}{\spe_i} + \tc \right ) \left (w_i \spr_i^2 + \ec \right )
\]
}{
\[\ex{\energy} \!=\!
\sum_i \!\left (w_i \spe_i^2 +\! \ec \right ) + \lambda\! \sum_i \! \left (\frac{w_i}{\spe_i} +\! \tc \right )\! \left (w_i \spr_i^2 +\! \ec \right )
\]
}
   subject to 
   \[\sum_i \left (\frac{w_i}{\spe_i} + \tc \right ) + \lambda \sum_i  \left (\frac{w_i}{\spe_i} + \tc \right ) \! \left (\frac{w_i}{\spr_i} + \tc \right ) \leq \dl\] in the \ed model, and subject to
\[ \sum_i \left (\frac{w_i}{\spe_i} + \tc \right )  +  \sum_i \left (\frac{w_i}{\spr_i} + \tc \right ) \leq \dl \] in the \hd model.
We first deal
with the \singles model, where we enforce $\spr_i = \spe_i$, before
dealing with the \multis model.

\subsection{Single speed model}
\subsubsection{Expected deadline}
In this section, we deal with the \singles \ed model and consider that for all $i$, $\spr_i=\spe_i$. Then:
\RRme{
\begin{align*}
\ex{\tps}(\cup_i (w_i,\spe_i,\spe_i)) &= \sum_i \left (\frac{w_i}{\spe_i} + \tc \right ) + \lambda \sum_i  \left (\frac{w_i}{\spe_i} + \tc \right )^2\\
\ex{\energy}(\cup_i (w_i,\spe_i,\spe_i))&=\sum_i \left (w_i \spe_i^2 + \ec \right )\left (1 + \lambda\left (\frac{w_i}{\spe_i} + \tc \right ) \right )
\end{align*}
}{
\begin{align*}
\ex{\tps}(\cup_i (w_i,\spe_i,\spe_i))\! &=\! \sum_i\! \left (\frac{w_i}{\spe_i}\! +\! \tc \right )\! +\! \lambda \sum_i \! \left (\frac{w_i}{\spe_i}\! +\! \tc \right )^2\\
\ex{\energy}(\cup_i (w_i,\spe_i,\spe_i))\!&=\!\sum_i\! \left (w_i \spe_i^2 \!+\! \ec \right )\!\left (1 \!+\! \lambda\left (\frac{w_i}{\spe_i}\! +\! \tc \right ) \! \right )
\end{align*}
}
\begin{theorem}
\label{th.gopi}
  In  the optimal solution to the problem with the \singles \ed model,
  all $n$ chunks are of
  equal size $W/n$ and executed at the same speed \spe.
\end{theorem}

\begin{proof}
  Consider the optimal solution, and assume by contradiction that it includes two chunks $w_1$ and $w_2$,
  executed at speeds  $\spe_1$ and $\spe_2$, where  either $\spe_1 \neq\spe_2$, or
  $\spe_1 =\spe_2$ and $w_1 \neq w_2$.
  Let us assume without loss of generality that $\frac{w_1}{\spe_1}
  \geq \frac{w_2}{\spe_2}$.

 We show that we can find a strictly better solution where both chunks have size
  $w=\frac{1}{2}(w_1+w_2)$, and are executed at same speed~$\spe$ (to
  be defined later). The size and speed of the other chunks are kept the same. 
  We will show that the execution time of the new solution is not larger 
  than in the optimal solution, while its energy consumption is strictly smaller,
  hence leading to the  
  contradiction.

We have seen that
\RRme{
\begin{align*}
\ex{\tps}((w_1,\spe_1),(w_2,\spe_2)) &= \frac{w_1}{\spe_1} +\tc + \frac{w_2}{\spe_2} +\tc + \lambda\left ( \frac{w_1}{\spe_1} + \tc\right )^2+ \lambda\left ( \frac{w_2}{\spe_2} + \tc\right )^2 \\
\ex{\tps}((w,\spe),(w,\spe)) &= 2\left (\frac{w}{\spe} +\tc \right )+ 2\lambda\left ( \frac{w}{\spe} + \tc\right )^2
\end{align*}
}{
\begin{align*}
\ex{\tps}&((w_1,\spe_1),(w_2,\spe_2)) = \frac{w_1}{\spe_1} +\tc + \frac{w_2}{\spe_2} +\tc \\  
& + \lambda\left ( \frac{w_1}{\spe_1} + \tc\right )^2+ \lambda\left ( \frac{w_2}{\spe_2} + \tc\right )^2 \\
\ex{\tps}&((w,\spe),(w,\spe)) = \\ 
& 2\left (\frac{w}{\spe} +\tc \right )+ 2\lambda\left ( \frac{w}{\spe} + \tc\right )^2
\end{align*}}
Hence,
\RRme{
\[\ex{\tps}((w_1,\spe_1),(w_2,\spe_2))\!-\!\ex{\tps}((w,\spe),(w,\spe)) = \left (\frac{w_1}{\spe_1}\!+\! \frac{w_2}{\spe_2}\!-\!\frac{2w}{\spe} \right )\!+\!
\lambda \left ( \left (\frac{w_1}{\spe_1}\right )^2\!+\!\left ( \frac{w_2}{\spe_2}\right )^2\!-\!2\left (\frac{w}{\spe}\right )^2\right )
\]
}{
\begin{align*}
&\ex{\tps}((w_1,\spe_1),(w_2,\spe_2))\!-\!\ex{\tps}((w,\spe),(w,\spe)) = \\  
&\left (\frac{w_1}{\spe_1}\!+\! \frac{w_2}{\spe_2}\!-\!\frac{2w}{\spe} \right )\!+\!
\lambda \left ( \left (\frac{w_1}{\spe_1}\right )^2\!+\!\left ( \frac{w_2}{\spe_2}\right )^2\!-\!2\left (\frac{w}{\spe}\right )^2\right )
\end{align*}}
Similarly, we know that:
\RRme{
\begin{align*}
\ex{\energy}((w_1,\spe_1),(w_2,\spe_2)) &= w_1\spe_1^2 +\ec + w_2\spe_2^2 +\ec 
+\lambda \left ( \frac{w_1}{\spe_1}\!+\!\tc\right )\!\left ( w_1\spe_1^2\!+\!\ec\right ) 
+\lambda\left ( \frac{w_2}{\spe_2}\!+\!\tc\right )\!\left ( w_2 \spe_2^2\!+\!\ec\right )\\
\ex{\energy}((w,\spe),(w,\spe)) &=  2\left (w\spe^2 +\ec \right )  +2\lambda\left ( \frac{w}{\spe} + \tc\right )\left ( w\spe^2 + \ec\right )
\end{align*}
}{
\begin{align*}
\ex{\energy}&((w_1,\spe_1),(w_2,\spe_2)) = w_1\spe_1^2 +\ec + w_2\spe_2^2 +\ec \\
&+\lambda \left ( \frac{w_1}{\spe_1}\!+\!\tc\right )\!\left ( w_1\spe_1^2\!+\!\ec\right ) \\
&+\lambda\left ( \frac{w_2}{\spe_2}\!+\!\tc\right )\!\left ( w_2 \spe_2^2\!+\!\ec\right )\\
\ex{\energy}&((w,\spe),(w,\spe)) =  2\left (w\spe^2 +\ec \right ) \\
& +2\lambda\left ( \frac{w}{\spe} + \tc\right )\left ( w\spe^2 + \ec\right )
\end{align*}}
and deduce
\RRme{
\begin{align}
	\label{eq.energy.blou}
\ex{\energy}&((w_1,\spe_1),(w_2,\spe_2)) -\ex{\energy}((w,\spe),(w,\spe)) \nonumber \\
&=  \left (w_1 \spe_1^2 + w_2\spe_2^2 - 2w\spe^2 \right ) \left ( 1+\lambda \tc\right ) 
+ \lambda \ec \left(\frac{w_1}{\spe_1} + \frac{w_2}{\spe_2} - \frac{2w}{\spe} \right)
+ \lambda \left ( w_1^2 \spe_1 + w_2^2 \spe_2 - 2w^2 \spe \right )
\end{align}
}{
\begin{align}
	\label{eq.energy.blou}
\ex{\energy}&((w_1,\spe_1),(w_2,\spe_2))
-\ex{\energy}((w,\spe),(w,\spe)) = \nonumber \\ 
& \left (w_1 \spe_1^2 + w_2\spe_2^2 - 2w\spe^2 \right ) \left ( 1+\lambda \tc\right ) \nonumber\\
&+ \lambda \ec \left(\frac{w_1}{\spe_1} + \frac{w_2}{\spe_2} - \frac{2w}{\spe} \right)\nonumber\\
&+ \lambda \left ( w_1^2 \spe_1 + w_2^2 \spe_2 - 2w^2 \spe \right )
\end{align}}

Let us now define
\begin{align*}
\spe_A &= \frac{2w}{\frac{w_1}{\spe_1} + \frac{w_2}{\spe_2}} = \frac{w_1 + w_2}{\frac{w_1}{\spe_1} + \frac{w_2}{\spe_2}}\\
\spe_B &=\!\frac{\sqrt{2}w}{\left ( \left (\frac{w_1}{\spe_1}\right )^2\!+\!\left ( \frac{w_2}{\spe_2}\right )^2\right )^{1/2}}\!=\!\frac{w_1+w_2}{\left ( 2\left (\frac{w_1}{\spe_1}\right )^2\!+\!2\left ( \frac{w_2}{\spe_2}\right )^2\right )^{1/2}}  
\end{align*}

We then fix $\spe = \max ( \spe_A, \spe_B)$. Then, since $\spe \geq
\spe_A$, we have $\frac{w_1}{\spe_1}\!+\!\frac{w_2}{\spe_2}\!-\!
\frac{2w}{\spe} \geq 0$, and since $\spe \geq \spe_B$, we have
$\left (\frac{w_1}{\spe_1}\right )^2\!+\!\left (\frac{w_2}{\spe_2}\right )^2\!-\!2\left ( \frac{w}{\spe}\right )^2 \geq
0$. This ensures that $\ex{\tps}((w_1,\spe_1),(w_2,\spe_2))-\ex{\tps}((w,\spe),(w,\spe)) \geq 0$. 

Note that
\RRme{ 
\begin{align*}
\frac{\left (w_1+w_2\right )^2}{\spe_B^2} &- \frac{\left (w_1+w_2\right )^2}{\spe_A^2} 
=  2\left (\frac{w_1}{\spe_1}\right )^2 + 2\left ( \frac{w_2}{\spe_2}\right )^2 - \left (\frac{w_1}{\spe_1} +\frac{w_2}{\spe_2}\right )^2 
 = \left (\frac{w_1}{\spe_1} -\frac{w_2}{\spe_2}\right )^2 \geq 0
\end{align*}
 }{
\begin{align*}
\frac{\left (w_1+w_2\right )^2}{\spe_B^2} &- \frac{\left (w_1+w_2\right )^2}{\spe_A^2} \\
&=  2\left (\frac{w_1}{\spe_1}\right )^2 + 2\left ( \frac{w_2}{\spe_2}\right )^2 - \left (\frac{w_1}{\spe_1} +\frac{w_2}{\spe_2}\right )^2 \\
& = \left (\frac{w_1}{\spe_1} -\frac{w_2}{\spe_2}\right )^2 \geq 0
\end{align*}}
This means that $\spe_A \geq \spe_B$, hence $\spe = \spe_A$.
To prove that $\ex{\energy}((w_1,\spe_1),(w_2,\spe_2))
-\ex{\energy}((w,\spe),(w,\spe)) > 0$, we want to show that:
\begin{enumerate}
	\item $w_1 \spe_1^2 + w_2\spe_2^2 - 2w\spe^2 \geq 0$
	\item $\frac{w_1}{\spe_1} + \frac{w_2}{\spe_2} - \frac{2w}{\spe} \geq 0$
	\item $w_1^2 \spe_1 + w_2^2 \spe_2 - 2w^2 \spe \geq 0$
	\item and that one of the previous inequalities is strict.
\end{enumerate}
Note that by definition of $\spe=\spe_A$, the second inequality is
true.

\smallskip
\paragraph{Let us first show that $ w_1 \spe_1^2 + w_2\spe_2^2 - 2w\spe_A^2 \geq 0$}
\RRme{
\begin{align*}
\left (\frac{w_1}{\spe_1} + \frac{w_2}{\spe_2}\right )^2&\left (w_1 \spe_1^2 + w_2\spe_2^2 - (w_1 + w_2)\left( \frac{w_1 + w_2}{\frac{w_1}{\spe_1} + \frac{w_2}{\spe_2}} \right)^2  \right ) \\
&\quad=w_1^3 + w_1w_2^2\left ( \frac{\spe_1}{\spe_2}\right )^2 + 2\frac{w_1w_2}{\spe_1 \spe_2}w_1\spe_1^2 +w_2^3 
+ w_2w_1^2\left ( \frac{\spe_2}{\spe_1}\right )^2 + 2\frac{w_2w_1}{\spe_1 \spe_2}w_2\spe_2^2  - \left ( w_1 + w_2\right )^3 \\
&\quad = w_1w_2^2 \left ( \left ( \frac{\spe_1}{\spe_2}\right )^2 +2\frac{\spe_2}{\spe_1} -3 \right ) + w_1^2w_2 \left ( \left ( \frac{\spe_2}{\spe_1}\right )^2 +2\frac{\spe_1}{\spe_2} -3 \right ) \\
&\quad = w_1w_2^2 g\!\!\left(\frac{\spe_1}{\spe_2}\right) + w_1^2w_2 g\!\!\left(\frac{\spe_2}{\spe_1}\right)
\end{align*}
}{
\begin{align*}
&\left (\frac{w_1}{\spe_1}\! +\! \frac{w_2}{\spe_2}\right )^2\!\left (w_1 \spe_1^2\! +\! w_2\spe_2^2 \!-\! (w_1\! +\! w_2)\left( \frac{w_1 \!+\! w_2}{\frac{w_1}{\spe_1}\!+\! \frac{w_2}{\spe_2}} \right)^2  \!\right ) \\
&\quad=w_1^3 + w_1w_2^2\left ( \frac{\spe_1}{\spe_2}\right )^2 + 2\frac{w_1w_2}{\spe_1 \spe_2}w_1\spe_1^2 +w_2^3 \\
& \quad \quad+ w_2w_1^2\left ( \frac{\spe_2}{\spe_1}\right )^2 + 2\frac{w_2w_1}{\spe_1 \spe_2}w_2\spe_2^2  - \left ( w_1 + w_2\right )^3 \\
&\quad = w_1w_2^2 g\!\!\left(\frac{\spe_1}{\spe_2}\right) + w_1^2w_2 g\!\!\left(\frac{\spe_2}{\spe_1}\right)
\end{align*}}
where $g: u \mapsto u^2 +\frac{2}{u} -3$. It is easy to show that $g$ is nonnegative on 
$\mathbb{R}^{\star}_+$: indeed, $g'(u) = \frac{2}{u^2}(u^{3}-1)$ is negative in $[0,1[$ and positive in $]1,\infty[$, and the unique minimum is $g(1)=0$.
We derive that  $w_1 \spe_1^2 +
w_2\spe_2^2 - 2w\spe_A^2 \geq 0$. 

\smallskip
\paragraph{Let us now show that $w_1^2 \spe_1 + w_2^2 \spe_2 - 2w^2 \spe \geq 0$}
Remember that $2w =w_1 + w_2$.
\RRme{
\begin{align*}
2\left (\frac{w_1}{\spe_1} + \frac{w_2}{\spe_2}\right )&\left (w_1^2 \spe_1 + w_2^2\spe_2 -  \frac{\left(w_1 + w_2 \right)^3}{2\left (\frac{w_1}{\spe_1} + \frac{w_2}{\spe_2}\right )}  \right ) \\
&=2w_1^3 + 2w_1^2w_2 \frac{\spe_1}{\spe_2}+ 2w_2^3 + 2w_1w_2^2 \frac{\spe_2}{\spe_1} - \left ( w_1 + w_2\right )^3 \\
&= w_1^3 + w_2^3 + w_1^2w_2 \left ( 2\frac{\spe_1}{\spe_2} -3\right )+ w_1w_2^2 \left ( 2\frac{\spe_2}{\spe_1} -3 \right )
\end{align*}
}{
\begin{align*}
2&\left (\frac{w_1}{\spe_1} + \frac{w_2}{\spe_2}\right )\!\left (w_1^2 \spe_1 + w_2^2\spe_2 -  \frac{\left(w_1 + w_2 \right)^3}{2\left (\frac{w_1}{\spe_1} + \frac{w_2}{\spe_2}\right )}  \right ) \\
&=2w_1^3 + 2w_1^2w_2 \frac{\spe_1}{\spe_2}+ 2w_2^3 + 2w_1w_2^2 \frac{\spe_2}{\spe_1} - \left ( w_1 + w_2\right )^3 \\
&= w_1^3 + w_2^3 + w_1^2w_2 \left ( 2\frac{\spe_1}{\spe_2} -3\right )+ w_1w_2^2 \left ( 2\frac{\spe_2}{\spe_1} -3 \right )
\end{align*}}
Remember that we assumed without loss of generality that 
$\frac{w_1}{\spe_1} \geq \frac{w_2}{\spe_2}$.
\RRme{
\begin{align*}
2\left (\frac{w_1}{\spe_1} + \frac{w_2}{\spe_2}\right ) &\left (w_1^2 \spe_1 + w_2^2\spe_2 -  \frac{\left(w_1 + w_2 \right)^3}{2\left (\frac{w_1}{\spe_1} + \frac{w_2}{\spe_2}\right )}  \right ) \\
\geq&w_2^3 \!\left (\!\left (\!\frac{\spe_1}{\spe_2}\!\right )^3\!+\!1\!+\!\left (\!\frac{\spe_1}{\spe_2}\!\right )^2\!\left (\!2\frac{\spe_1}{\spe_2}\! -\!3\!\right )\!+\!\frac{\spe_1}{\spe_2}\! \left (\! 2\frac{\spe_2}{\spe_1}\! -\!3\! \right )\! \right )\\
\geq& 3 w_2^3\!\left (\!\left (\!\frac{\spe_1}{\spe_2}\!\right )^3\!-\!\left (\!\frac{\spe_1}{\spe_2}\!\right )^2\!-\!\frac{\spe_1}{\spe_2}\!+\!1\!\right) \\
\geq &3 w_2^3\!\left (\!\left (\!\frac{\spe_1}{\spe_2}\!-\!1\!\right )^2\!\left (\!\frac{\spe_1}{\spe_2}\!+\!1\!\right )\!\right) \geq 0
\end{align*}
}{
\begin{align*}
2&\left (\frac{w_1}{\spe_1} + \frac{w_2}{\spe_2}\right ) \left (w_1^2 \spe_1 + w_2^2\spe_2 -  \frac{\left(w_1 + w_2 \right)^3}{2\left (\frac{w_1}{\spe_1} + \frac{w_2}{\spe_2}\right )}  \right ) \\
\geq&w_2^3 \!\left (\!\left (\!\frac{\spe_1}{\spe_2}\!\right )^3\!+\!1\!+\!\left (\!\frac{\spe_1}{\spe_2}\!\right )^2\!\left (\!2\frac{\spe_1}{\spe_2}\! -\!3\!\right )\!+\!\frac{\spe_1}{\spe_2}\! \left (\! 2\frac{\spe_2}{\spe_1}\! -\!3\! \right )\! \right )\\
\geq& 3 w_2^3\!\left (\!\left (\!\frac{\spe_1}{\spe_2}\!\right )^3\!-\!\left (\!\frac{\spe_1}{\spe_2}\!\right )^2\!-\!\frac{\spe_1}{\spe_2}\!+\!1\!\right) \\
\geq &3 w_2^3\!\left (\!\left (\!\frac{\spe_1}{\spe_2}\!-\!1\!\right )^2\!\left (\!\frac{\spe_1}{\spe_2}\!+\!1\!\right )\!\right) \geq 0
\end{align*}}

Let us now conclude our study:  if
$\frac{\spe_1}{\spe_2} \neq 1$, then the energy consumption of the
optimal solution is strictly greater than the one from our solution
which is a contradiction. Hence we must have $\spe_1=\spe_2$, and  $w_1 \neq w_2$ (in
fact, since we assumed that $\frac{w_1}{\spe_1} \geq
\frac{w_2}{\spe_2}$, we must have $w_1>w_2$). Then we can refine the previous analysis, and
obtain that $w_1^2 \spe_1 + w_2^2 \spe_2 - 2w^2
\spe > 0$: again, the optimal energy
consumption is strictly greater than in our solution; this is the final
contradiction and concludes the proof.
\end{proof}

\smallskip
Thanks to this result, we know that the problem with $n$ chunks can be rewritten as follows:
find \spe such that 
\RRme{
\begin{align*}
n\left(\frac{W}{n\spe} + \tc\right ) + n\lambda\left (\frac{W}{n\spe} +\tc\right )^2  = \frac{W}{\spe} + n\tc + \frac{\lambda}{n}\left (\frac{W}{\spe} +n\tc\right )^2 \leq \dl
\end{align*}
}{
\begin{align*}
n&\left(\frac{W}{n\spe} + \tc\right ) + n\lambda\left (\frac{W}{n\spe} +\tc\right )^2 \\
&\quad = \frac{W}{\spe} + n\tc + \frac{\lambda}{n}\left (\frac{W}{\spe} +n\tc\right )^2 \leq \dl
\end{align*}}
in order to minimize
\RRme{
\begin{align*}
n\left(\frac{W}{n}\spe^2 + \ec\right )+ n\lambda\left (\frac{W}{n\spe} +\tc\right )\left ( \frac{W}{n}\spe^2 +\ec\right )  = \left (W\spe^2 + n\ec\right )\left ( 1 + \frac{\lambda}{n}\left (\frac{W}{\spe} +n\tc\right ) \right)
\end{align*}
}{
\begin{align*}
n&\left(\frac{W}{n}\spe^2 + \ec\right )+ n\lambda\left (\frac{W}{n\spe} +\tc\right )\left ( \frac{W}{n}\spe^2 +\ec\right ) \\
&\quad = \left (W\spe^2 + n\ec\right )\left ( 1 + \frac{\lambda}{n}\left (\frac{W}{\spe} +n\tc\right ) \right)
\end{align*}}

One can see that this reduces to the \singlec problem with the \singles model (Section~\ref{sec.singles-one})
up to the following parameter changes:
\begin{itemize}
	\item $\lambda \leftarrow \frac{\lambda}{n}$
	\item $\tc \leftarrow n\tc$
	\item $\ec \leftarrow n\ec$
\end{itemize}


If the number of chunks~$n$ is given, we can express the minimum speed such that
there is a solution with $n$ chunks: 
\begin{equation}
\spe_0(n) = W\dfrac{1+2\lambda \tc + \sqrt{4 \frac{\lambda \dl}{n} +1}}{2(\dl -n\tc(1+\lambda \tc))}.
\end{equation}

We can verify that when $\dl \leq n\tc ( 1 + \lambda n)$, there is no solution, 
hence obtaining an upper bound on~$n$.
Therefore, the two variables problem (with unknowns $n$ and $s$) can be solved numerically.

\subsubsection{Hard deadline}
In the \hd model, all results still hold, they are even easier to prove since we do not need to 
introduce a second speed.
\begin{theorem}
  In  the optimal solution to the problem with the \singles \hd model,
  all $n$ chunks are of
  equal size $W/n$ and executed at the same speed \spe.
\end{theorem}

\begin{proof}
The proof is similar to the one of Theorem~\ref{th.gopi}, except we do not need to study the 
case where $\spe_B > \spe_A$.
\end{proof}

\subsection{Multiple speeds model}

\subsubsection{Expected deadline}
In this section, we still deal with the problem of a divisible task of size $W$ 
 that we can split into an arbitrary number of chunks, but using  the more general \multis model.  We start
by proving that all re-execution speeds are equal: 

\RRme{
Let us first introduce a preliminary Lemma:
\begin{lemma}[Convexity \multic]
	\label{conv.tight.div}
The problem of minimizing $A_0 + \alpha_0 x_0^2 + \alpha_1 x_1^2$ under the constraint $\frac{\alpha_0}{x_0} +\frac{\alpha_1}{x_1} \leq A_1$
where $A_0$ is a constant, and $A_1, \alpha_0, \alpha_1$ are positive constants, is solved 
when $x_0 = x_1$, and when the constraint is tight: $\frac{\alpha_0}{x_0} +\frac{\alpha_1}{x_1} = A_1$.
\end{lemma}

\begin{proof}
First remark that when $x_1$ is fixed, then according to Lemma~\ref{conv.tight.undiv}, the constraint 
should be tight. Hence this is true for the optimal solution (any optimal solution when the constraint is
not tight can be improved by reducing one of the variables).

To prove the result now that we know that the constraint is tight, it suffices to replace in the function 
we wish to minimize, $x_0 = \frac{\alpha_0}{A_1 - \frac{\alpha_1}{x_1}}$.
Differentiating $A_0 + \alpha_0 \times \left ( \frac{\alpha_0}{A_1 - \frac{\alpha_1}{x_1}}\right )^2 + \alpha_1 x_1^2$
with respect to $x_1$ gives $-\frac{2\alpha_1 \alpha_0^3}{x_1^2 \left (A_1 - \frac{\alpha_1}{x_1} \right )^3} + 2\alpha_1 x_1$.
Then we obtain that the equation is minimized (by differentiating again, we can see that the function 
is convex) when $-\frac{2\alpha_1 \alpha_0^3}{x_1^2 \left (A_1 - \frac{\alpha_1}{x_1} \right )^3} + 2\alpha_1 x_1 = 0$,
that is $-x_0 + x_1 = 0$, hence the result.
\end{proof}
Note that if $A_1$ is nonpositive, then there is no solution.
}{}

\begin{lemma}
\label{lemma.samereexec}
  In the \multis model,  all
  re-execution speeds are equal in the optimal solution: $\exists \spr, \forall i, \spr_i=\spr$, and 
  the deadline is tight.
\end{lemma}
\RRme{
\begin{proof}
This is a direct corollary of Lemma~\ref{conv.tight.div}.
  If we consider the $w_i$ and $\spe_i$ to be fixed, then we can write
  $\ex{\tps}(\cup_i (w_i,\spe_i,\spr_i)) = \tps_0 + \sum_i
  \frac{\alpha_i}{\spr_i}$, and $\ex{\energy}(\cup_i (w_i,\spe_i,\spr_i)) =
  \energy_0 + \sum_i \alpha_i \spr_i^2$, where $\tps_0$, $\energy_0$
  and $\alpha_i$ are constant.  Assuming $\dl-\tps_0 > 0$ (otherwise there is no solution),
  we can apply Lemma~\ref{conv.tight.div}, then the problem is minimized
  when the deadline is tight, and when for all $i$, $\spr_i = \frac{\sum_i \alpha_i}{\dl - T_0}$.
\end{proof}
}{
Due to lack of space, the proof is available in the companion research report~\cite{rr-abmrr}.
}

We can now redefine 
\begin{align*}
\ex{\tps}(\cup_i (w_i,\spe_i,\spr_i)) &= \tps(\cup_i (w_i,\spe_i),\spr) \\
\ex{\energy}(\cup_i (w_i,\spe_i,\spr_i)) &= \energy(\cup_i (w_i,\spe_i),\spr)
\end{align*}

\begin{theorem}
  In the \multis model, all
  chunks have the same size $w_i=\frac{W}{n}$, and are executed at the same speed $\spe$, 
  in the optimal solution.
 \end{theorem}

\RRme{
\begin{proof}
  We first prove that chunks are of equal size. Assume first, by contradiction, that the
  optimal solution has two chunks of different sizes, for instance $w_1
  < w_2$. These chunks are executed at speeds $\spe_1$
  and~$\spe_2$. Thanks to Lemma~\ref{lemma.samereexec}, both chunks
  are re-executed at a same speed~$\spr$.
  We consider the solution with two chunks of size
  $w=\frac{1}{2}(w_1+w_2)$, executed at a same speed~$\spe$ (to
  be defined later), and re-executed at speed~$\spr$ (the value of the
  re-execution speed in the optimal solution). The size and speed of
  the other chunks are kept the same. We show that 
   the execution time  is not greater
  than in the optimal solution, while the energy consumption is
  strictly smaller, hence leading to the contradiction.

We have seen that
\RRme{
\begin{align*}
\ex{\tps}((w_1,\spe_1),(w_2,\spe_2),\spr) &=  \frac{w_1}{\spe_1} +\tc + \frac{w_2}{\spe_2} +\tc +
\lambda\left ( \frac{w_1}{\spe_1} + \tc\right )\left ( \frac{w_1}{\spr} + \tc\right )+
\lambda\left ( \frac{w_2}{\spe_2} + \tc\right )\left ( \frac{w_2}{\spr} + \tc\right ) \\
\ex{\tps}((w,\spe),(w,\spe),\spr) &=  2\left (\frac{w}{\spe} +\tc \right )+ 
2\lambda\left ( \frac{w}{\spe} + \tc\right )\left ( \frac{w}{\spr} + \tc\right )
\end{align*}
}{
\begin{align*}
\ex{\tps}&((w_1,\spe_1),(w_2,\spe_2),\spr) \\  
=&\frac{w_1}{\spe_1} +\tc + \frac{w_2}{\spe_2} +\tc \\
+&\lambda\!\left (\!\frac{w_1}{\spe_1}\! +\!\tc\right )\!\left (\!\frac{w_1}{\spr}\!+\!\tc\right )\!+\!
\lambda\!\left (\! \frac{w_2}{\spe_2}\! +\! \tc\right )\!\left ( \!\frac{w_2}{\spr}\! +\! \tc\right ) \\
\ex{\tps}&((w,\spe),(w,\spe),\spr) \\
=& 2\left (\frac{w}{\spe} +\tc \right )+ 
2\lambda\left ( \frac{w}{\spe} + \tc\right )\left ( \frac{w}{\spr} + \tc\right )
\end{align*}}
Hence,
\RRme{
\begin{align*}
\ex{\tps}((w_1,\spe_1),(w_2,\spe_2),\spr)  -\ex{\tps}((w,\spe),(w,\spe),\spr) &=  
 \left ( 1 +\lambda \tc \right ) \left (\frac{w_1}{\spe_1} + \frac{w_2}{\spe_2} - \frac{2w}{\spe} \right )+
\frac{\lambda}{\spr} \left ( \frac{w_1^2}{\spe_1} + \frac{w_2^2}{\spe_2} - \frac{2w^2}{\spe}\right )
\end{align*}
}{
\begin{align*}
\ex{\tps}&((w_1,\spe_1),(w_2,\spe_2),\spr)  -\ex{\tps}((w,\spe),(w,\spe),\spr) \\
=&\left (\!1\! +\!\lambda\tc \!\right )\! \left (\!\frac{w_1}{\spe_1}\! +\! \frac{w_2}{\spe_2}\! -\! \frac{2w}{\spe}\! \right )\!+\!
\frac{\lambda}{\spr}\! \left (\! \frac{w_1^2}{\spe_1}\! +\! \frac{w_2^2}{\spe_2}\! -\! \frac{2w^2}{\spe}\!\right )
\end{align*}}
Similarly, 
we know that:
\RRme{
\begin{align*}
\ex{\energy}((w_1,\spe_1),(w_2,\spe_2),\spr) &=  w_1\spe_1^2 +\ec + w_2\spe_2^2 +\ec +
\lambda\left ( \frac{w_1}{\spe_1} + \tc\right )\left ( w_1\spr^2 + \ec\right )+
\lambda\left ( \frac{w_2}{\spe_2} + \tc\right )\left ( w_2 \spr^2 + \ec\right ) \\
\ex{\energy}((w,\spe),(w,\spe),\spr) &=  2\left (w\spe^2 +\ec \right )+ 
2\lambda\left ( \frac{w}{\spe} + \tc\right )\left ( w\spr^2 + \ec\right )
\end{align*}
}{
\begin{align*}
\ex{\energy}&((w_1,\spe_1),(w_2,\spe_2),\spr)\\
 =&  w_1\spe_1^2 \!+\!\ec\! +\! w_2\spe_2^2\! +\!\ec\! \\
 +&\lambda\!\left (\! \frac{w_1}{\spe_1}\! +\! \tc\!\right )\!\left ( \!w_1\spr^2\! +\! \ec\!\right )\!+\!
\lambda\!\left ( \!\frac{w_2}{\spe_2}\! +\! \tc\!\right )\!\left ( \!w_2 \spr^2\! + \!\ec\!\right ) \\
\ex{\energy}&((w,\spe),(w,\spe),\spr) \\
=&  2\left (w\spe^2 +\ec \right )+ 
2\lambda\left ( \frac{w}{\spe} + \tc\right )\left ( w\spr^2 + \ec\right )
\end{align*}}
and deduce
\RRme{
\begin{align}
	\label{eq.energy.blibla}
\ex{\energy}((w_1,\spe_1),(w_2,\spe_2),\spr)&
-\ex{\energy}((w,\spe),(w,\spe),\spr) = \nonumber \\ 
& \left (w_1 \spe_1^2 + w_2\spe_2^2 - 2w\spe^2 \right )
+ \lambda \ec \left(\frac{w_1}{\spe_1} + \frac{w_2}{\spe_2} - \frac{2w}{\spe} \right)
+ \lambda \spr^2 \left ( \frac{w_1^2}{\spe_1} + \frac{w_2^2}{\spe_2} - \frac{2w^2}{\spe}\right )
\end{align}
}{
\begin{align}
	\label{eq.energy.blibla}
\ex{\energy}&((w_1,\spe_1),(w_2,\spe_2),\spr)
-\ex{\energy}((w,\spe),(w,\spe),\spr)  \nonumber \\ 
&= \left (w_1 \spe_1^2 + w_2\spe_2^2 - 2w\spe^2 \right )
+ \lambda \ec\! \left(\frac{w_1}{\spe_1} + \frac{w_2}{\spe_2} - \frac{2w}{\spe} \right) \nonumber\\
&+ \lambda \spr^2 \left ( \frac{w_1^2}{\spe_1} + \frac{w_2^2}{\spe_2} - \frac{2w^2}{\spe}\right )
\end{align}
}

Let us now define
\begin{align*}
\spe_A &= \frac{2w}{\frac{w_1}{\spe_1} + \frac{w_2}{\spe_2}} = \frac{w_1 + w_2}{\frac{w_1}{\spe_1} + \frac{w_2}{\spe_2}}\\
\spe_B &=  \frac{2w^2}{\frac{w_1^2}{\spe_1} + \frac{w_2^2}{\spe_2}} =  \frac12 \frac{(w_1+w_2)^2}{\frac{w_1^2}{\spe_1} + \frac{w_2^2}{\spe_2}}  
\end{align*}

We then fix $\spe = \max ( \spe_A, \spe_B)$. Then, since $\spe \geq
\spe_A$, we have $\frac{w_1}{\spe_1} + \frac{w_2}{\spe_2} -
\frac{2w}{\spe} \geq 0$, and since $\spe \geq \spe_B$, we have
$\frac{w_1^2}{\spe_1} + \frac{w_2^2}{\spe_2} - \frac{2w^2}{\spe} \geq
0$. This ensures that $\ex{\tps}((w_1,\spe_1),(w_2,\spe_2),\spr)
-\ex{\tps}((w,\spe),(w,\spe),\spr) \geq 0$. To prove that
$\ex{\energy}((w_1,\spe_1),(w_2,\spe_2),\spr)
-\ex{\energy}((w,\spe),(w,\spe),\spr) \geq 0$, there remains to show that 
$w_1 \spe_1^2 + w_2\spe_2^2 - 2w\spe^2 \geq 0$.
 %
 %

\paragraph{Let us first suppose that $\spe_A > \spe_B$}
Then we have $\spe=\spe_A$, and
let us show that $ w_1 \spe_1^2 + w_2\spe_2^2 - 2w\spe_A^2 \geq 0$: 
\RRme{
\begin{align*}
\left (\frac{w_1}{\spe_1} + \frac{w_2}{\spe_2}\right )^2 & \left (w_1 \spe_1^2 + w_2\spe_2^2 - (w_1 + w_2)\left( \frac{w_1 + w_2}{\frac{w_1}{\spe_1} + \frac{w_2}{\spe_2}} \right)^2  \right ) \\
=&w_1^3 + w_1w_2^2\left ( \frac{\spe_1}{\spe_2}\right )^2 + 2\frac{w_1w_2}{\spe_1 \spe_2}w_1\spe_1^2 +w_2^3 + w_2w_1^2\left ( \frac{\spe_2}{\spe_1}\right )^2 + 2\frac{w_2w_1}{\spe_1 \spe_2}w_2\spe_2^2  - \left ( w_1 + w_2\right )^3 \\
=& w_1w_2^2 \left ( \left ( \frac{\spe_1}{\spe_2}\right )^2 +2\frac{\spe_2}{\spe_1} -3 \right ) + w_1^2w_2 \left ( \left ( \frac{\spe_2}{\spe_1}\right )^2 +2\frac{\spe_1}{\spe_2} -3 \right ) \\
=& w_1w_2^2 g\!\!\left(\frac{\spe_1}{\spe_2}\right) + w_1^2w_2 g\!\!\left(\frac{\spe_2}{\spe_1}\right)
\end{align*}
}{
\begin{align*}
&\left (\frac{w_1}{\spe_1} + \frac{w_2}{\spe_2}\right )^2\! \left (w_1 \spe_1^2 + w_2\spe_2^2 - (w_1 + w_2)\left( \frac{w_1 + w_2}{\frac{w_1}{\spe_1} + \frac{w_2}{\spe_2}} \right)^2  \right ) \\
&\quad = w_1^3 + w_1w_2^2\left ( \frac{\spe_1}{\spe_2}\right )^2 + 2\frac{w_1w_2}{\spe_1 \spe_2}w_1\spe_1^2 +w_2^3 \\
 &\quad \quad + w_2w_1^2\left ( \frac{\spe_2}{\spe_1}\right )^2 + 2\frac{w_2w_1}{\spe_1 \spe_2}w_2\spe_2^2  - \left ( w_1 + w_2\right )^3 \\
&\quad = w_1w_2^2 \! \left (\! \left ( \!\frac{\spe_1}{\spe_2}\right )^2\! +\!2\frac{\spe_2}{\spe_1}\! -\!3\! \right )\! 
+\! w_1^2w_2\! \left (\! \left ( \!\frac{\spe_2}{\spe_1}\!\right )^2\! +\!2\frac{\spe_1}{\spe_2} \!-\!3\! \right ) \\
&\quad = w_1w_2^2 g\!\!\left(\frac{\spe_1}{\spe_2}\right) + w_1^2w_2 g\!\!\left(\frac{\spe_2}{\spe_1}\right)
\end{align*}}
where $g: u \mapsto u^2 +\frac{2}{u} -3$. We know from the proof of Theorem~\ref{th.gopi}
that $g$ is positive on 
$\mathbb{R}^{\star}_+$,  hence  $w_1 \spe_1^2 +
w_2\spe_2^2 - 2w\spe_A^2 \geq 0$. 

Finally, since $\spe > \spe_B$, we have
$\frac{w_1^2}{\spe_1} + \frac{w_2^2}{\spe_2} - \frac{2w^2}{\spe} > 0$,
and all other terms of 
$\ex{\energy}((w_1,\spe_1),(w_2,\spe_2),\spr)
-\ex{\energy}((w,\spe_A),(w,\spe_A),\spr)$ are non-negative, hence proving
that the new solution is strictly better than the optimal one, and
leading to a contradiction. 

\paragraph{Let us now suppose that $\spe_A \leq \spe_B$}
Then we have $\spe=\spe_B$. Moreover, we have 
$(w_2- w_1)(\frac{w_2}{\spe_2} - \frac{w_1}{\spe_1}) \leq 0$ (this
comes directly from $\spe_A \leq \spe_B$), and since we
assume that $w_2 > w_1$, 
$\frac{w_2}{\spe_2} - \frac{w_1}{\spe_1} \leq 0$. 
%
Let us show that $w_1 \spe_1^2 + w_2\spe_2^2 - 2w\spe_B^2 > 0$: 

\begin{align*}
4\left (\frac{w^2_1}{\spe_1} + \frac{w^2_2}{\spe_2}\right )^2 &\times \left (w_1 \spe_1^2 + w_2\spe_2^2 - (w_1 + w_2)\left( \frac12 \frac{(w_1 + w_2)^2}{\frac{w^2_1}{\spe_1} + \frac{w^2_2}{\spe_2}} \right)^2  \right ) \\
=& 4w_1^5 + 8w_1^3w_2^2 \frac{\spe_1}{\spe_2} + 4 w_1w_2^4\left ( \frac{\spe_1}{\spe_2}\right )^2 +4w_2^5 + 8w_1^2w_2^3 \frac{\spe_2}{\spe_1} + 4 w_1^4w_2\left ( \frac{\spe_2}{\spe_1}\right )^2 - \left ( w_1 + w_2\right )^5 \\
= & 3\left (w_1^5 + w_2^5 \right ) + w_1^3 w_2^2 \left ( 8\frac{\spe_1}{\spe_2} - 10\right ) + w_2^3 w_1^2 \left ( 8\frac{\spe_2}{\spe_1} - 10\right )\\
&\quad + w_1 w_2^4 \left ( 4\left (\frac{\spe_1}{\spe_2}\right )^2 -5 \right ) + w_1^4 w_2 \left ( 4\left (\frac{\spe_2}{\spe_1}\right )^2 -5 \right )\\
\end{align*}
\begin{align*}
4&\!\left (\!\frac{w^2_1}{\spe_1}\! +\! \frac{w^2_2}{\spe_2}\!\right )^2 \! \left (\!w_1 \spe_1^2\! + \!w_2\spe_2^2\! -\! (w_1\! +\! w_2)\!\left(\! \frac12 \!\frac{(\!w_1\! +\! w_2\!)^2}{\frac{w^2_1}{\spe_1}\! +\! \frac{w^2_2}{\spe_2}}\! \right)^2\!  \right ) \\
&= 4w_1^5 \!+\! 8w_1^3w_2^2 \frac{\spe_1}{\spe_2}\! +\! 4 w_1w_2^4\!\left ( \frac{\spe_1}{\spe_2}\right )^2 \!+\!4w_2^5 \!+\! 8w_1^2w_2^3 \frac{\spe_2}{\spe_1} \\
&\quad + 4 w_1^4w_2\left ( \frac{\spe_2}{\spe_1}\right )^2 - \left ( w_1 + w_2\right )^5 \\
 &= 3\!\left (\!w_1^5\! +\! w_2^5\! \right )\! +\! w_1^3 w_2^2\! \left (\! 8\frac{\spe_1}{\spe_2}\! - \!10\!\right )\! +\! w_2^3 w_1^2\! \left (\! 8\frac{\spe_2}{\spe_1} \!-\! 10\!\right )\\
&\quad + \! w_1 w_2^4\! \left (\! 4\!\left (\!\frac{\spe_1}{\spe_2}\!\right )^2 \!-\!5\! \right )\! +\! w_1^4 w_2 \!\left ( \!4\!\left (\!\frac{\spe_2}{\spe_1}\!\right )^2\! -\!5\! \right )
\end{align*} 
Now because $w_1 \geq \frac{w_2\spe_1}{\spe_2}$, we can bound
the last equation. Let $u=\frac{\spe_1}{\spe_2}$ (and
hence $w_1\geq u\times w_2$):
\begin{align*}
4\left (\frac{w^2_1}{\spe_1} + \frac{w^2_2}{\spe_2}\right )^2 &\times
\left (w_1 \spe_1^2 + w_2\spe_2^2 - (w_1 + w_2)\left( \frac12
    \frac{(w_1 + w_2)^2}{\frac{w^2_1}{\spe_1} + \frac{w^2_2}{\spe_2}}
  \right)^2  \right ) \\
\geq & w_2^5 \left (3\left ( u^5 +1 \right ) + u^3 \left( 8 u -10 \right ) + u^2\left ( 8 \frac{1}{u} -10 \right )+ u \left ( 4u^2 -5\right )+ u^4\left ( 4\frac{1}{u^2} -5\right ) \right ) \\
 = & w_2^5 \left (3u^5 + 3u^4 - 6u^3 -6u^2 +3u +3 \right )\\
 =& 3 w_2^5 \left ( u-1\right )^2 \left (u+1 \right )^3
\end{align*}
\begin{align*}
4&\!\left (\!\frac{w^2_1}{\spe_1}\! +\! \frac{w^2_2}{\spe_2}\!\right )^2 \! \left (\!w_1 \spe_1^2\! + \!w_2\spe_2^2\! -\! (w_1\! +\! w_2)\!\left(\! \frac12 \!\frac{(\!w_1\! +\! w_2\!)^2}{\frac{w^2_1}{\spe_1}\! +\! \frac{w^2_2}{\spe_2}}\! \right)^2\!  \right ) \\
&\geq w_2^5 \!\left (\!3\!\left (\! u^5 \!+\!1\! \right )\! +\! u^3 \!\left(\! 8 u \!-\!10\! \right )\! + u^2\!\left ( \!8 \frac{1}{u} \!-\!10 \!\right ) \right . \\
&\quad \left .+ u\! \left (\! 4u^2\! -\!5\!\right )\!+\! u^4\!\left ( \!4\frac{1}{u^2} \!-\!5\!\right ) \! \right ) \\
 &=  w_2^5 \left (3u^5 + 3u^4 - 6u^3 -6u^2 +3u +3 \right )\\
 &= 3 w_2^5 \left ( u-1\right )^2 \left (u+1 \right )^3
\end{align*} 
Since $w_2>w_1$, $0<u<1$, and this polynomial is strictly positive,
hence we have   
$w_1 \spe_1^2 + w_2\spe_2^2 - 2w\spe_B^2 > 0$.

Finally, we can conclude that in both cases, 
$\ex{\energy}((w_1,\spe_1),(w_2,\spe_2),\spr)
-\ex{\energy}((w,\spe_B),(w,\spe_B),\spr) > 0$, so there exist a better
solution with two chunks of same sizes, hence leading to a
contradiction.

We had proven that all chunks have the same size. We use the same line of reasoning
 to prove that all chunks are executed at a same
speed~$\spe$. If there are two chunks executed
at speeds $\spe_1<\spe_2$ (with $w_1=w_2=w$), then we have
$s_A=s_B$. Considering that $s=s_A$, it is easy to see that $w_1
\spe_1^2 + w_2\spe_2^2 - 2w\spe_A^2 > 0$ since $w_1w_2^2
g\!\!\left(\frac{\spe_1}{\spe_2}\right) + w_1^2w_2
g\!\!\left(\frac{\spe_2}{\spe_1}\right)>0$. Indeed, $g$ is null only
in~$1$, and $\spe_1 \neq \spe_2$. We exhibit a solution strictly
better, hence showing a contradiction. This concludes the proof. 
\end{proof}
}
{Due to lack of space, the proof is available in the companion research report~\cite{rr-abmrr}.
This proof uses the same reasoning as the proof of Theorem~\ref{th.gopi}.}





Thanks to this result, we know that the $n$ chunks problem can be rewritten as follows:
find \spe such that 
\begin{itemize}
	\item $\frac{W}{\spe} + n\tc + \frac{\lambda}{n}\left (\frac{W}{\spe} +n\tc\right )\left (\frac{W}{\spr} +n\tc\right ) = \dl$
	\item in order to minimize $W\spe^2 + n\ec + \frac{\lambda}{n}\left (\frac{W}{\spe} +n\tc\right )\left (W\spr^2 +n\ec\right )$
\end{itemize}

One can see that this reduces to the \singlec \multis \ed task problem where 
\begin{itemize}
	\item $\lambda \leftarrow \frac{\lambda}{n}$
	\item $\tc \leftarrow n\tc$
	\item $\ec \leftarrow n\ec$
\end{itemize}
and allows us to write the problem to solve as a two parameters function:
\RRme{
\begin{equation}
	\label{eq.multis.div}
(n,\spe) \mapsto W \spe^2 +n\ec + \frac{\lambda}{n} \left(  \frac{W}{\spe} + n\tc\right ) \left (W \left (\frac{\frac{\lambda}{n} W}{\frac{\dl}{ \frac{W}{\spe} + n\tc} - (1 + \lambda \tc)} \right )^2 + n\ec \right)
\end{equation}
}{
\begin{align}
	\label{eq.multis.div}
(n,\spe) \! & \mapsto  W \spe^2 +n\ec \nonumber \\
&\!\!\!\!+\! \frac{\lambda}{n} \!\left(\!  \frac{W}{\spe} \!+\! n\tc\!\right )\! \left (\!W\! \left (\!\frac{\frac{\lambda}{n} W}{\frac{\dl}{ \frac{W}{\spe}\! +\! n\tc}\! -\! (\!1\! +\! \lambda \tc\!)} \!\right )^2\! +\! n\ec\! \right)
\end{align}}
which can be minimized numerically.

\medskip
\subsubsection{Hard deadline}
In this section, the constraint on the execution time can be written as: 
\[\sum_i \left ( \frac{w_i}{\spe_i} + \tc + \frac{w_i}{\spr_i} + \tc \right ) \leq \dl.\]

\begin{lemma}
	\label{lemma.tight.hd.div}
In the \multis \hd model with divisible chunk, the deadline should be tight.
\end{lemma}
\RRme{
\begin{proof}
This result is obvious with Lemma~\ref{conv.tight.undiv}: if we have a solution such that the 
deadline is not tight, if we fix every variable but $\spr_1$ (the re-execution speed of the first task), we
can improve the solution with a tight deadline.
\end{proof}
}{}
\begin{lemma}
	\label{lemma.spr.hd.div}
In the optimal solution, for all $i,j$, $\lambda \left ( \frac{w_i}{\spe_i} + \tc \right ) \spr_i^3 = \lambda \left ( \frac{w_j}{\spe_j} + \tc \right ) \spr_j^3 $.
\end{lemma}
\RRme{
\begin{proof}
Consider any solution to our problem. Thanks to Lemma~\ref{lemma.tight.hd.div}, we know that the 
deadline should be tight. Let $T_i$ and $T_j$ two tasks of er-execution speed $\spr_i,\spr_j$. We 
show that those speed can be optimally defined such that 
$\lambda \left ( \frac{w_i}{\spe_i} + \tc \right ) \spr_i^3 = \lambda \left ( \frac{w_j}{\spe_j} + \tc \right ) \spr_j^3 $.
Let us call $u_i = \lambda \left ( \frac{w_i}{\spe_i} + \tc \right )$ and $u_j = \lambda \left ( \frac{w_j}{\spe_j} + \tc \right )$.

The minimization problem for those speeds can be written as $A_0 +  u_i w_i \spr_i^2 + u_j w_j \spr_j^2$
under the constraint that $A_1 +\frac{w_i}{\spr_i} + \frac{w_j}{\spr_j} = \dl$ where neither
$A_0$ nor $A_1$ depends on $\spr_i,\spr_j$.

Replacing $\spr_i \! =\! \frac{w_i}{\dl - A_1 \! - \! \frac{w_j}{\spr_j}}$ in the function we need to 
minimize, we obtain $A_0 +  u_i w_i \left ( \frac{w_i}{\dl - A_1 - \frac{w_j}{\spr_j}}\right )^2 + u_j w_j \spr_j^2$.
A simple differentiation gives $-2 w_j u_i\frac{w_i^3}{\left (\dl - A_1 - \frac{w_j}{\spr_j} \right )^3 \spr_j^2} + 2 u_j w_j \spr_j$.
Another differentiation shows the convexity of the function we want to minimize.
Hence one can see that the function is minimized when $u_j \spr_j^3 = u_i \left (\frac{w_i}{\dl - A_1 - \frac{w_j}{\spr_j}} \right )^3 = u_i \spr_i^3$.
\end{proof}
}{The proofs for both lemmas are available in the companion research report~\cite{rr-abmrr}.}

\begin{lemma}
\label{lemma.weakgopi}
If we enforce the condition that the execution speeds of the chunks are all equal, and that the re-execution speeds
of the chunks are all equal, then all chunks should have same size  in the optimal solution.
\end{lemma}

\begin{proof}
This result is obvious since the problem can be reformulated as the minimization of 
$\alpha \sum w_i + \beta \sum w_i^2$ where neither $\alpha$ nor $\beta$ depends on any 
$w_i$, under the constraints $\gamma \sum w_i + \zeta \leq \dl$, and 
$\sum w_i = W$. It is easy to see the result when there are only two chunks since there is only one 
variable, and the problem generalizes well in the case of $n$ chunks.
\end{proof}




We have not been able to prove a stronger result than Lemma~\ref{lemma.weakgopi}. However
we conjecture the following result:
\begin{conjecture}
In the \multis \hd, in the optimal solution, the re-execution speeds are identical, the deadline is tight.
The re-execution speed is  equal to $\spr=\frac{W}{(D-2n\tc)\spe - W} \spe$.
Furthermore the chunks should have the same size $\frac{W}{n}$ and should be executed at the same speed~\spe.
\end{conjecture}

This conjecture reduces the problem to the \singlec \multis problem where 
\begin{itemize}
	\item $\lambda \leftarrow \frac{\lambda}{n}$
	\item $\tc \leftarrow n\tc$
	\item $\ec \leftarrow n\ec$
\end{itemize}
and allows us to write the problem to solve as a two-parameter function:
\RRme{
\begin{equation}
	\label{eq.multis.div.hd}
(n,\spe) \mapsto W \spe^2 +n\ec +  \frac{\lambda}{n} \left(  \frac{W}{\spe} + n\tc\right ) \left (W \left (\frac{W}{(\dl-2n\tc)\spe - W} \spe \right )^2 + n\ec \right)
 \end{equation}
}{
\begin{align}
	\label{eq.multis.div.hd}
(n,\spe)\! \mapsto & W \spe^2 \!+\!n\ec \!+ \! \frac{\lambda}{n}\! \left( \! \frac{W}{\spe}\! +\! n\tc\!\right )\! \times \nonumber\\
& \quad\left (\!W\! \left (\!\frac{W}{(\!\dl\!-\!2n\tc\!)\!\spe\! -\! W} \spe\! \right )^2\! +\! n\ec\! \right)
 \end{align}}
 which can be solved numerically.

\section{Simulations}
\label{sec.experiments}

\RRme{
\begin{figure*}
\begin{center}
\includegraphics[width=\textwidth]{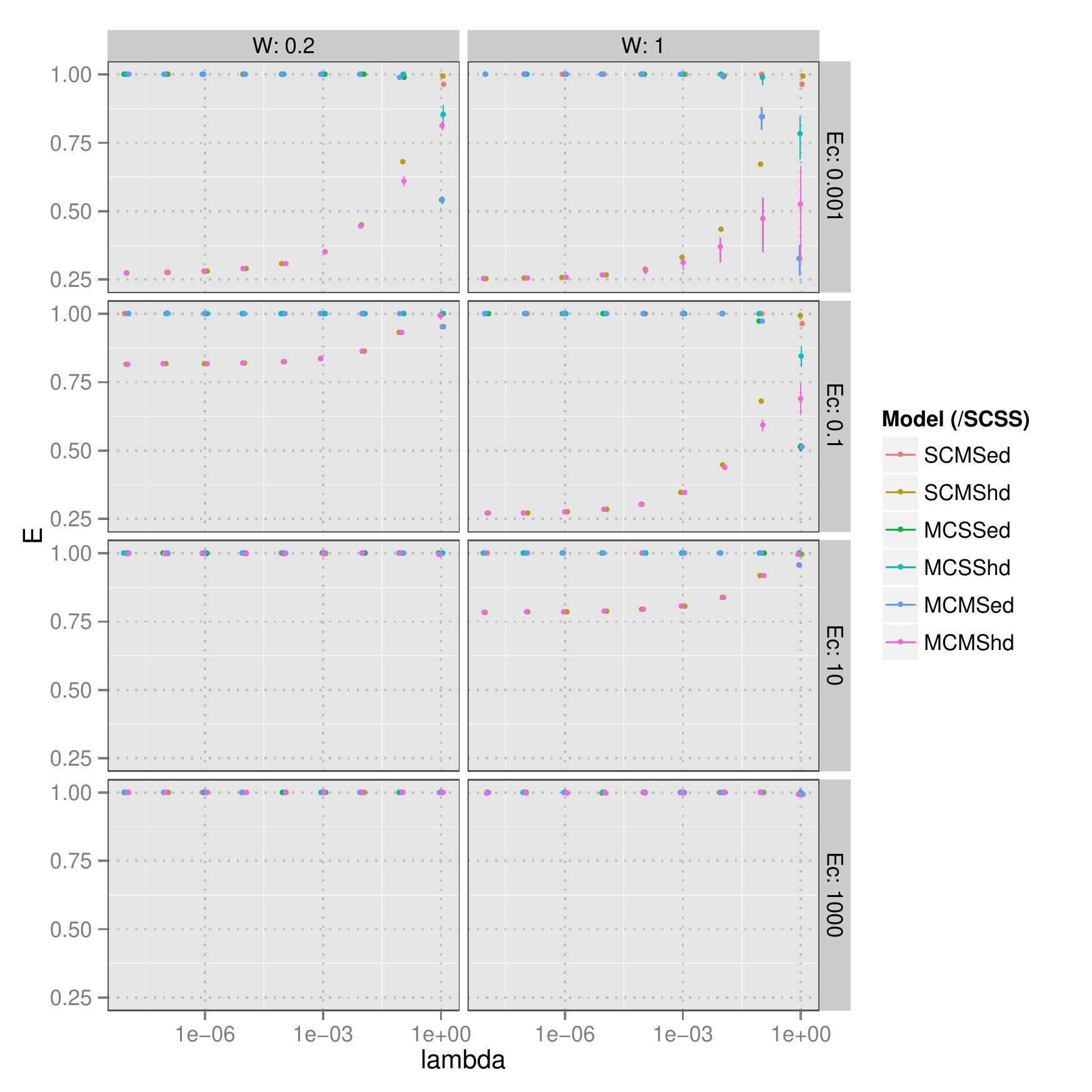}\\
\end{center}
\caption{Comparison with \singlec \singles.}
\label{fig.scss1}
\end{figure*}

\begin{figure*}
\begin{center}
\includegraphics[width=\textwidth]{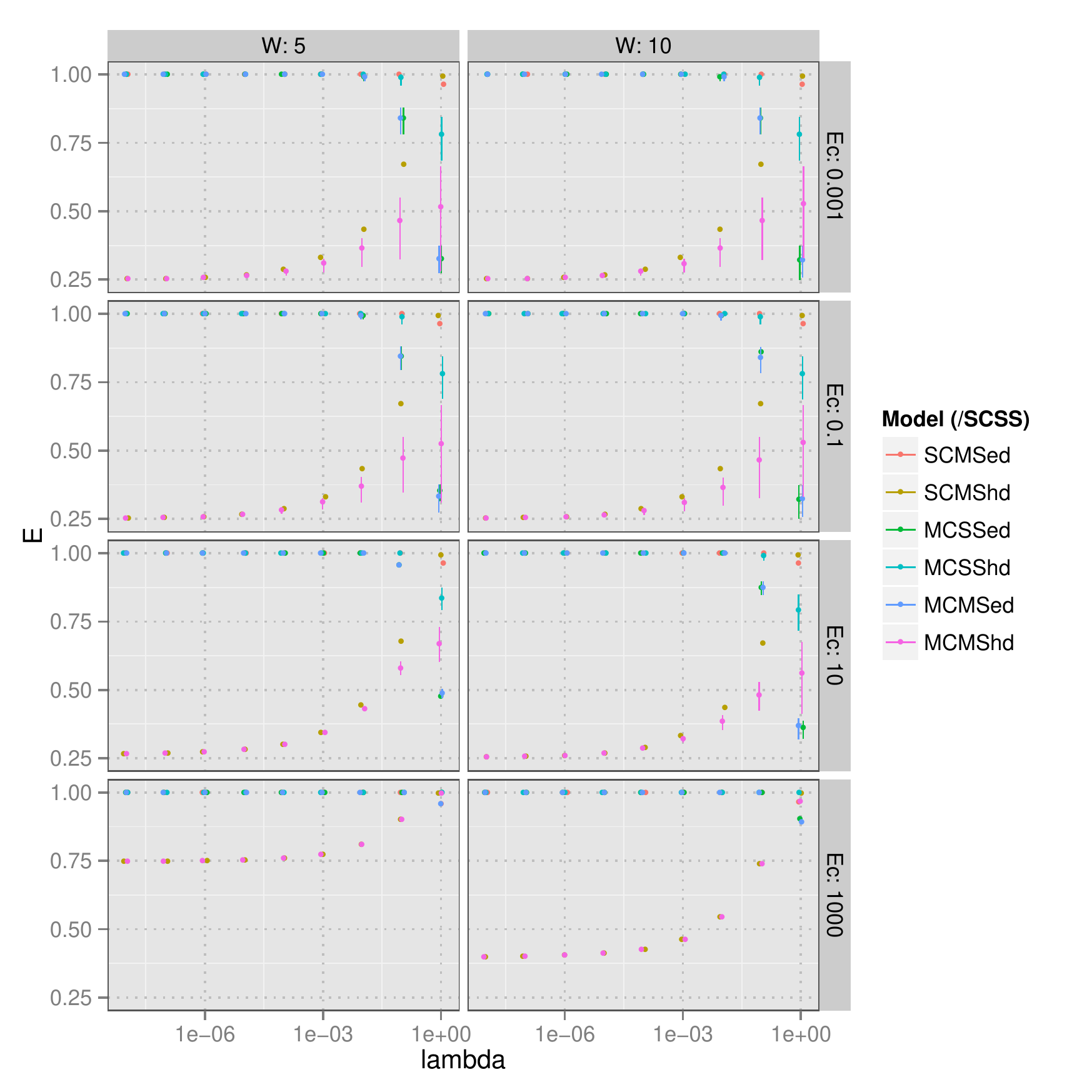}\\
\end{center}
\caption{Comparison with \singlec \singles.}
\label{fig.scss2}
\end{figure*}

\begin{figure*}
\begin{center}
\includegraphics[width=\textwidth]{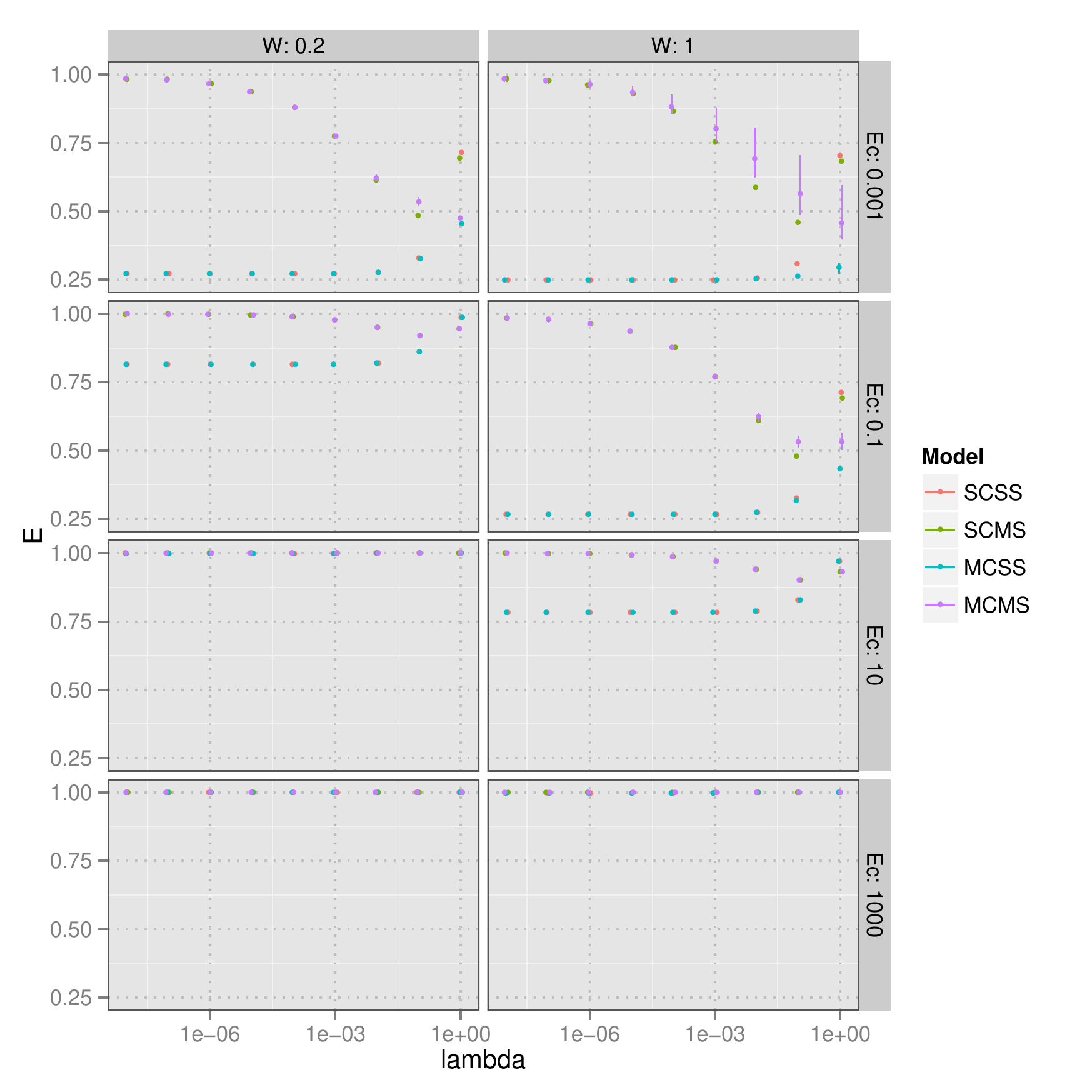}\\
\end{center}
\caption{Comparison \hd versus \ed.}
\label{fig.edhd1}
\end{figure*}

\begin{figure*}
\begin{center}
\includegraphics[width=\textwidth]{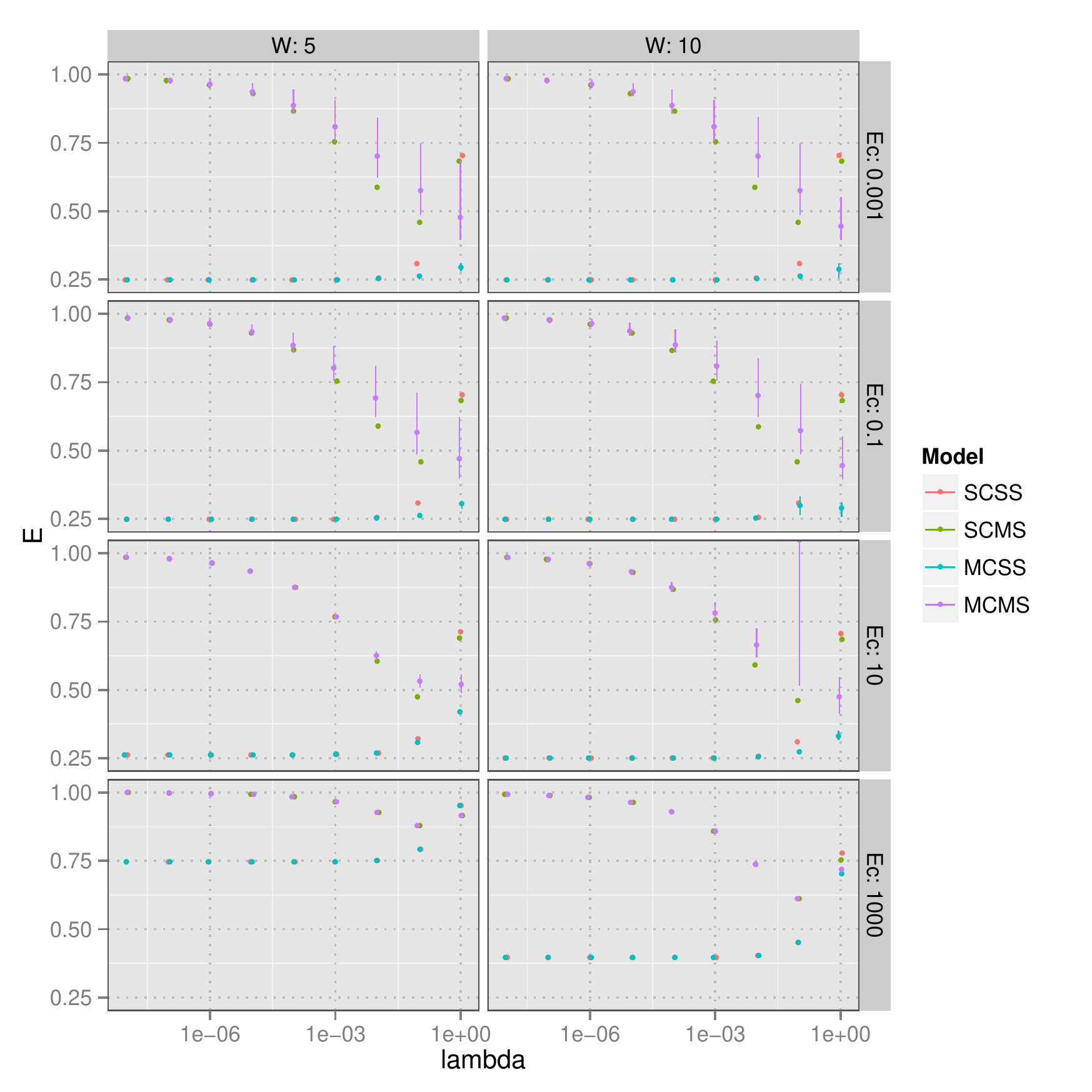}\\
\end{center}
\caption{Comparison \hd versus \ed.}
\label{fig.edhd2}
\end{figure*}
}{
\begin{figure*}
\begin{center}
\includegraphics[width=\textwidth]{conf-SCSS.pdf}\\
\end{center}
\caption{Comparison with single chunk single speed.}
\label{fig.scss}
\end{figure*}

\begin{figure*}
\begin{center}
\includegraphics[width=\textwidth]{conf-EDHD.pdf}\\
\end{center}
\caption{Comparison hard deadline versus expected deadline.}
\label{fig.edhd}
\end{figure*}
}

\subsection{Simulation settings}
We performed a large set of simulations in order to illustrate the differences between all the models studied
in this paper, and to show to which extent each additional degree of freedom improves the results, i.e., allowing for multiple
speeds instead of a single speed, or for multiple smaller chunks instead of a single large chunk. All these experiments
are conducted under both constraint types, expected and hard deadlines.

We envision reasonable settings by varying parameters within the following ranges:
\begin{itemize}
	\item $\frac{W}{\dl} \in [0.2,10]$
	\item $\frac{\tc}{\dl} \in [10^{-4},10^{-2}]$
	\item $\ec \in [10^{-3},10^{3}]$
	\item $\lambda \in [10^{-8},1]$.
\end{itemize}
In addition, we set the deadline to $1$. Note that since we study $\frac{W}{\dl}$ and 
$\frac{\tc}{\dl}$ instead of $W$ and \tc, we do not need to study how the variation of the deadline
impacts the simulation, this is already taken into account.

We use the Maple software to solve numerically the different minimization problems. Results are showed from two perspectives:
on the one hand (\RRme{Figures~\ref{fig.scss1} and~\ref{fig.scss2}}{Figure~\ref{fig.scss}}), for a given constraint (\hd or \ed), we normalize all variants according to
\RRme{\singles \singlec}{SCSS (single chunk single speed)}, under the considered constraint. 
For instance, on the plots, the energy consumed by \RRme{\multic \multis (denoted as MCMS)}{MCMS (multiple chunks multiple speeds)} 
for \hd \RRme{}{(hard deadline)} is divided by the energy consumed by \RRme{\singlec \singles (denoted as SCSS)}{SCSS} for \hd, while the energy of \RRme{\multic \singles
(denoted as MCSS)}{MCSS} for \ed is normalized by the energy of \RRme{\singlec \singles}{SCSS} for \ed.

On the other hand (\RRme{Figures~\ref{fig.edhd1}
and~\ref{fig.edhd2}}{Figure~\ref{fig.edhd}}), we study the impact of the constraint hardness on the energy consumption. For each solution form (\singles or
\multis, and \singlec or \multic), we plot the ratio energy consumed for \ed over energy consumed for \hd.

Note that for each figure, we plot for each function different values that depend on the 
different values of $\tc/\dl$ (hence the vertical intervals for points where $\tc/\dl$ has an impact).
In addition, the lower the value of $\tc/\dl$, the lower the energy consumption.

\subsection{Comparison with single speed}

At first, we observe that the results are identical for any value of $W/\dl$, up to a translation of \ec 
(see $(W/\dl=0.2, \ec=10^{-3})$ vs. $(W/\dl=5,\ec=1000)$ \RRme{}{or $(W/\dl=1,\ec=10^{-3})$ vs. $(W/\dl=5,\ec=0.1)$} on \RRme{Figures~\ref{fig.scss1} 
and~\ref{fig.scss2}}{Figure~\ref{fig.scss}}\RRme{, or see $(W/\dl=1,\ec=10^{-3})$ vs. $(W/\dl=5,\ec=0.1)$ on Figures~\ref{fig.scss1} 
and~\ref{fig.scss2}}{}, for instance).

Then the next observation is that for \ed, with a small $\lambda$ ($<10^{-2}$), \multic or 
\multis models do not improve the energy ratio. 
This is due to the fact that,
in both expressions for energy and for execution time, the re-execution term is negligible relative to the
execution one, since it has a weighting factor $\lambda$.
However, when $\lambda$ increases, if the energy of a 
checkpoint is  small relative to the total work (which is the  general case), we can see a huge 
improvement (between 25\% and 75\% energy saving) with \multic.


On the contrary, as expected, for small $\lambda$'s, re-executing at a different speed has a huge 
impact for \hd, where we can gain up to $75\%$ energy when the failure rate is low. 
We can indeed run at around half speed during the first execution (leading to the $1/2^2=25\%$ 
saving), and at a high speed for the second one, because the very low failure probability avoids the 
explosion of the expected energy consumption. For both \multic and \singlec, this saving
ratio increases with $\lambda$ (the energy consumed by the second execution cannot be neglected any more, and both executions 
need to be more balanced), the latter being more sensitive to $\lambda$. But the former is the only configuration where
$\tc$ has a significant impact: its performance decreases with $\tc$; still it remains strictly better 
than \RRme{\singlec \multis}{SCMS}.

\subsection{Comparison between \ed and \hd}

As before, the value of $W/\dl$ does not change the energy ratios up to translations of \ec. 
As expected, the difference between the \ed and \hd models is very important for the \singles 
variant: when the energy of the re-execution is negligible (because of the failure rate parameter), it would 
be better to spend as little time as possible doing the re-execution in order to have a speed as slow 
as possible for the first execution, however we are limited in the \singles \hd model by the fact that 
the re-execution time is fully taken into account (its speed is the same as the first execution, and
there is no parameter~$\lambda$ to render it negligible). 

Furthermore, when $\lambda$ is minimum, \multis consumes the same energy for \ed and for \hd. 
Indeed, as expected, the $\lambda$ in the energy function makes it possible for the
re-execution speed to be maximal: it has little impact on the energy, and it is optimal for the execution time;
this way we can focus on slowing down the first execution of each chunk. For \hd, we already run the
first execution at half speed, thus we cannot save more energy, even considering \ed instead.
When $\lambda$ increases, speeds of \hd cannot
be lowered but the expected execution time decreases, making room for a downgrade of the speeds in the \ed problems. 

\section{Conclusion}
\label{sec.conclusion}

In this work, we have studied the energy consumption of a divisible computational workload
on volatile 
platforms. In particular, we have studied the expected energy consumption under different deadline
constraints: a soft deadline (a deadline for the expected execution time), and a hard deadline (a
deadline for the worst case execution time).  
\RRme{

}{} 
We have been able to show mathematically, for all cases but one, that when using the
\RRme{\multic}{multiple chunks} model\RRme{}{ (MC)}, then 
(i) every chunk should be equally sized; (ii) every execution speed should be
equal; and (iii)  every re-execution speed should also be equal. 
This problem remains open in the
\RRme{\multis \hd}{multiple speeds hard deadline} variant.
\RRme{

}{  }  
Through a set of extensive simulations, we were able to show the following:  (i) when the fault parameter $\lambda$
is small, \RRme{for \ed constraints}{with expected deadline},  the \RRme{\singlec \singles}{single chunk single speed} model \RRme{}{(SCSS)} leads to almost optimal energy consumption. This is not true 
\RRme{for the \hd model}{with hard deadlines}, which accounts equally for execution and re-execution, thereby leading to higher
energy consumption.  Therefore, for the \hd 
model (hard deadline) and for small values of~$\lambda$, the model of choice  should be 
\RRme{the \singlec \multis model}{single chunk multiple speeds}, and that is not intuitive. 
When the fault parameter rate $\lambda$ increases, using a single chunk 
is no longer energy-efficient,
and one should focus on the \RRme{\multic \multis}{
MCMS} model for 
both deadline types.

An interesting direction for future work is to extend this study to the case of an application workflow: instead of dealing
with a single divisible task, we would deal with a DAG of tasks, that could be either divisible (checkpoints can take place
anytime) or atomic (checkpoints can only take place at the end of the execution of some tasks). Again, we can envision
both soft or hard constraints on the execution time, and we can keep the same model with a single re-execution per chunk/task,
at the same speed or possibly at a different speed. Deriving complexity results and heuristics to solve this
difficult problem is likely to be very challenging, but could have a dramatic impact to reduce the energy consumption of many
scientific applications.

\newpage

\bibliographystyle{abbrv}
\bibliography{biblio}

\end{document}